%% file: ms.tex
\documentclass{ieeetmlcn}

\usepackage{bbm}
\usepackage{amsmath,amssymb,amsfonts}
\usepackage{algorithm,algorithmic}
\usepackage{nicefrac}
\usepackage{bm}
\usepackage{hyperref}
\hypersetup{hidelinks=true} %
\usepackage{textcomp} %
\usepackage{xcolor} %
\usepackage{cite}
\usepackage{graphicx}
\usepackage[font=small]{caption}
\usepackage{tikz}
\usetikzlibrary{positioning}
\usepackage{pgfplots}
\usepackage{epstopdf}
\usepackage[labelformat=simple]{subcaption}
\usepackage{multirow}
\usepackage{lipsum}
\usepackage{mathtools}
\usepackage{tablefootnote}

\def\BibTeX{{\rm B\kern-.05em{\sc i\kern-.025em b}\kern-.08em
		T\kern-.1667em\lower.7ex\hbox{E}\kern-.125emX}}
\AtBeginDocument{\definecolor{tmlcncolor}{cmyk}{0.93,0.59,0.15,0.02}\definecolor{NavyBlue}{RGB}{0,86,125}}

\long\def\comment#1{}

\newtheorem{lemma}{Lemma}

\newtheorem{theorem}{Theorem}

\def\OJlogo{\vspace{-4pt}\includegraphics[height=18pt]{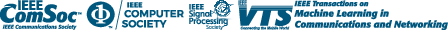}}
\def\seclogo{\vspace{10pt}\includegraphics[height=18pt]{new_logo}}

\def\authorrefmark#1{\ensuremath{^{\textbf{#1}}}}

\input{mysymbol.sty}

\newcommand{\update}[1]{#1}

\receiveddate{09 April, 2024}
\reviseddate{12 July, 2024}
\accepteddate{04 September, 2024}
\publisheddate{XX Month, XXXX}
\currentdate{XX Month, XXXX}
\doiinfo{TMLCN.2024.XXXXXXX}

\title{Biased Backpressure Routing Using Link Features and Graph Neural Networks}

\author{Zhongyuan Zhao\authorrefmark{1}, Member, IEEE, Bojan Radojičić\authorrefmark{2}, Gunjan Verma\authorrefmark{3},\\ Ananthram Swami\authorrefmark{3}, Life Fellow,~IEEE, and Santiago Segarra\authorrefmark{1}, Senior Member, IEEE}
\affil{Department of Electrical and Computer Engineering, Rice University, Houston, TX 77005 USA}
\affil{Faculty of Technical Sciences, University of Novi Sad, Novi Sad 21000, Serbia (Work was done as a visiting scholar at Rice University)}
\affil{US Army’s DEVCOM Army Research Laboratory (ARL), Adelphi, MD 20783 USA}
\corresp{Corresponding author: Zhongyuan Zhao (email: zhongyuan.zhao@rice.edu).}
\authornote{Research was sponsored by the DEVCOM ARL Army Research Office and was accomplished under Cooperative Agreement Numbers W911NF-19-2-0269 \update{and W911NF-24-2-0008}. 
	The views and conclusions contained in this document are those of the authors and should not be interpreted as representing the official policies, either expressed or implied, of the Army Research Office or the U.S. Government. 
	The U.S. Government is authorized to reproduce and distribute reprints for Government purposes notwithstanding any copyright notation herein. Preliminary results were presented in~\cite{zhao2023icassp,zhao2023enhanced}.
	The source code is available at \url{https://github.com/zhongyuanzhao/dutyBP}
}

\pgfplotsset{compat=1.16}
\begin{document}

\markboth{}{Zhao \MakeLowercase{\textit{et al.}}: Biased Backpressure Routing Using Link Features and Graph Neural Networks}


\begin{abstract}
To reduce the latency of Backpressure (BP) routing in wireless multi-hop networks, we propose to enhance the existing shortest path-biased BP (SP-BP) and sojourn time-based backlog metrics, since they introduce no additional time step-wise signaling overhead to the basic BP.
Rather than relying on hop-distance, we introduce a new edge-weighted shortest path bias built on the scheduling duty cycle of wireless links, which can be predicted by a graph convolutional neural network based on the topology and traffic of wireless networks.
Additionally, we tackle three long-standing challenges associated with SP-BP: optimal bias scaling, efficient bias maintenance, and integration of delay awareness. 
Our proposed solutions inherit the throughput optimality of the basic BP, as well as its practical advantages of low complexity and fully distributed implementation. 
Our approaches rely on common link features and introduces only a one-time constant overhead to previous SP-BP schemes, or a one-time overhead linear in the network size to the basic BP.
Numerical experiments show that our solutions can effectively address the major drawbacks of slow startup, random walk, and the last packet problem in basic BP, improving the end-to-end delay of existing low-overhead BP algorithms under various settings of network traffic, interference, and mobility.
\end{abstract}
\maketitle

\vspace{-0.1in}
\section{INTRODUCTION}\label{sec:intro}

Backpressure (BP) routing~\cite{tassiulas1992} is a well-established algorithm for distributed routing and scheduling in wireless multi-hop networks.
Wireless multi-hop networks have been widely adopted in military communications, disaster relief, and wireless sensor networks, and are envisioned to support emerging applications such as connected vehicles, robotic swarms, Internet of Things, machine-to-machine communications, and 5G/6G (device-to-device, wireless backhaul, integrated access and backhaul, and non-terrestrial coverage) \cite{Lin06,sarkar2013ad,kott2016internet,Patriciello2016,akyildiz20206g,cisco2020,chen2021massive,Cudak2021}.
The distributed nature of BP algorithms  \cite{tassiulas1992,neely2005dynamic,georgiadis2006resource,neely2022stochastic,jiao2015virtual,cui2016enhancing,gao2018bias,Alresaini2016bp,ji2012delay,athanasopoulou2012back,hai2018delay,Rai2017loop,yin2017improving,ying2010combining,Ying2012scheduling,Ryu2012timescale,zhao2023icassp,zhao2023enhanced} enables wireless multi-hop networks to be self-organized without relying on infrastructure, and promotes scalability and robustness that are critical to many applications.
BP has also been adopted to maximize the throughput of TCP networks~\cite{Liaskos2023}.
In BP, each node maintains a separate queue for packets to each destination (also denominated as commodity), routing decisions are made by selecting the commodity that maximizes the differential backlog between the two ends of each link, and data transmissions are activated on a set of non-interfering links via MaxWeight scheduling~\cite{tassiulas1992,joo2012local,zhao2021icassp}.
This mechanism uses congestion gradients to drive data packets towards their destinations through all possible routes, avoiding congestion and stabilizing the queues in the network for any flow rate within the network capacity region -- a property known as \emph{throughput optimality}~\cite{tassiulas1992,neely2005dynamic,georgiadis2006resource,neely2022stochastic}. 


\begin{figure}
    \vspace{-0.15in}
    \centerline{
    \subfloat[]{
    \includegraphics[height=1.2in]{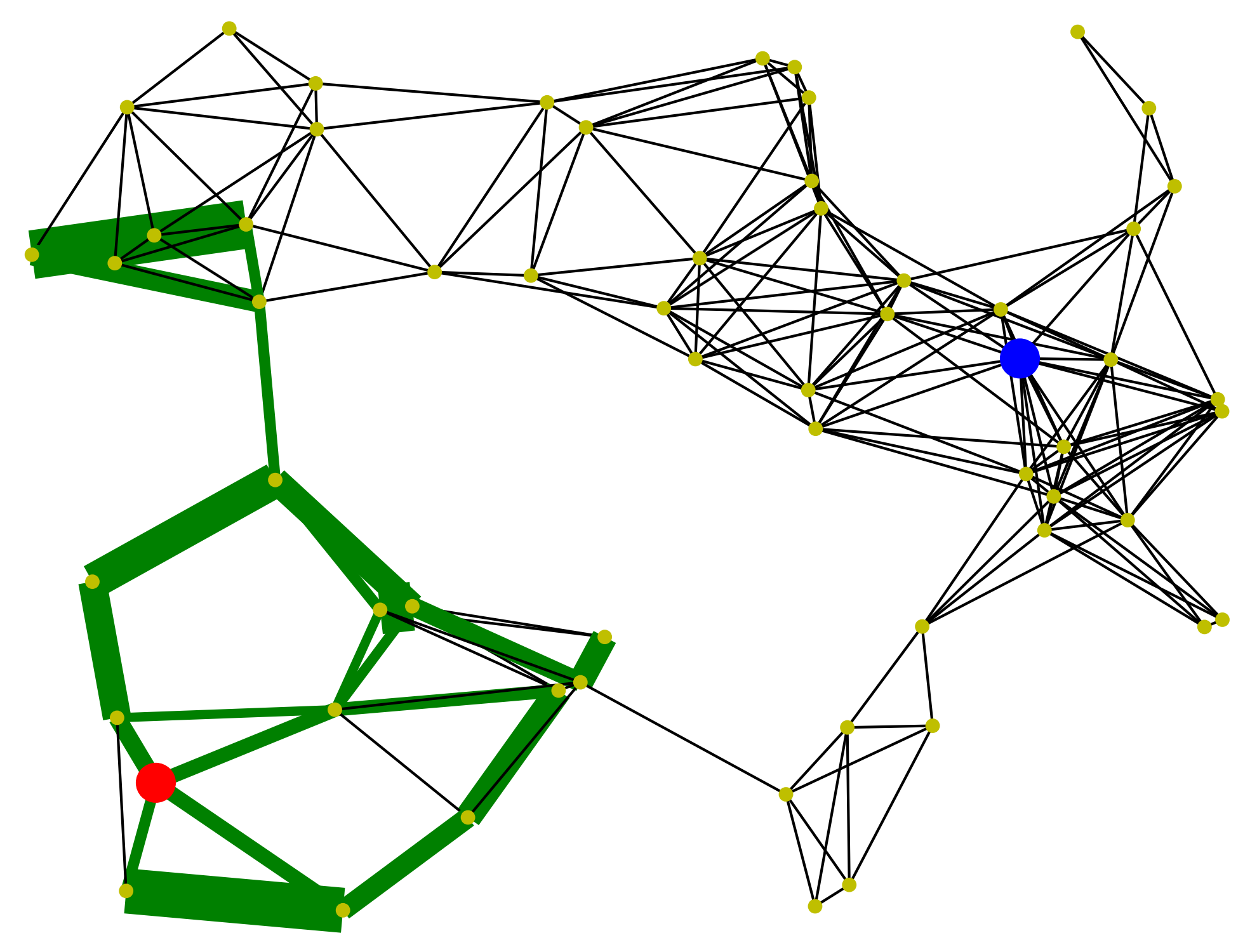}
    \label{fig:motivation:basic}\vspace{-0.1in}
    } 
    \subfloat[]{
    \includegraphics[height=1.2in]{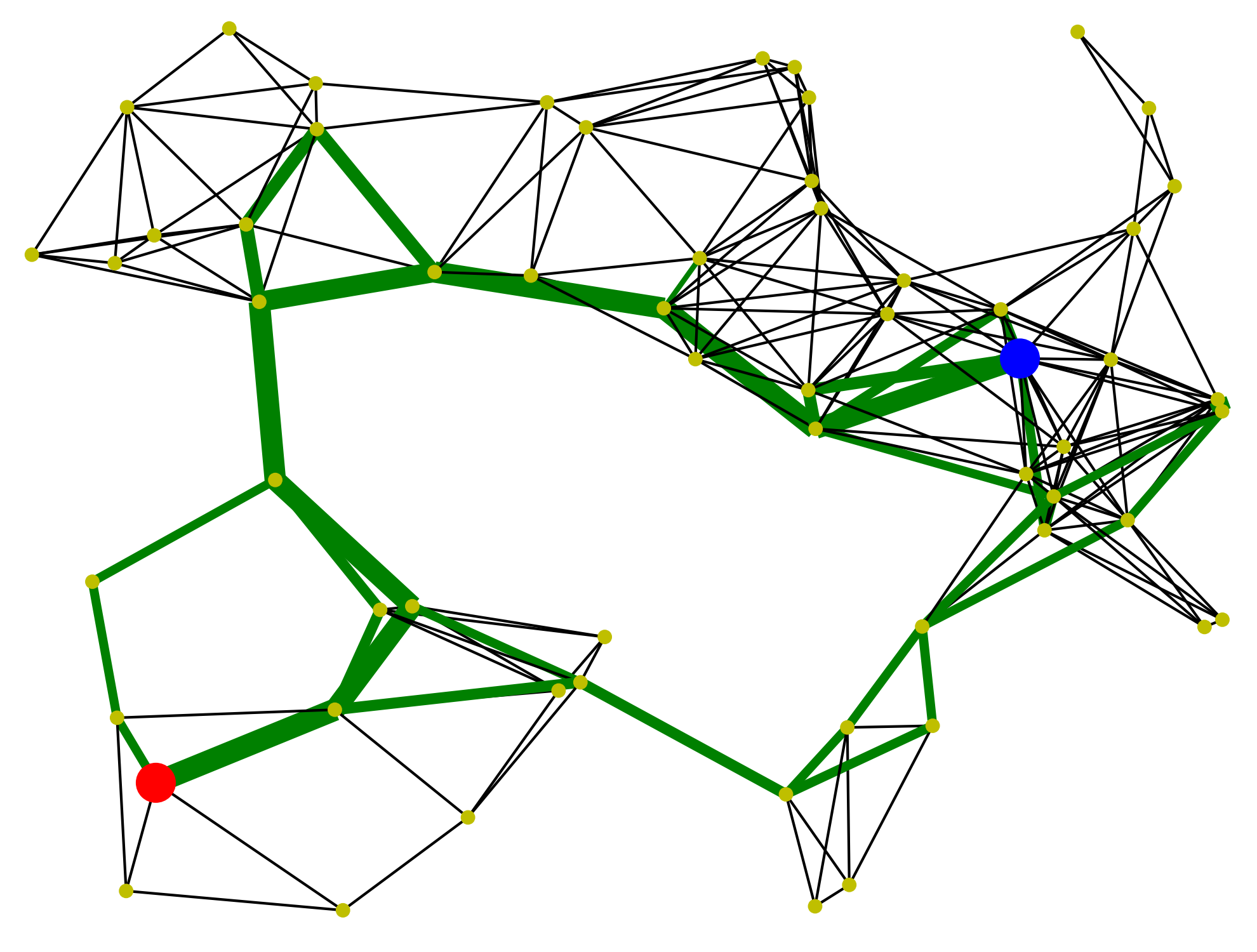}
    \label{fig:motivation:hop}\vspace{-0.1in}
    }
}
    \caption{A flow from the red node to the blue node in a wireless multi-hop network with 60 nodes. 
    The width of an edge is $1+\sqrt[3]{n}$, where $n$ is the number of packets sent over that link in 500 time steps; green edges indicate routes ($n>0$). 
    (a) Basic BP routing. (b) Enhanced dynamic BP routing (EDR) \cite{neely2005dynamic,georgiadis2006resource} with a pre-defined bias given by a scaled shortest hop distance from a node to the destination.    
    \label{fig:motivation}}
    \vspace{-0.1in}
\end{figure}

However, classical BP routing is known to suffer from poor delay performance, especially under low-to-medium traffic loads   \cite{neely2005dynamic,georgiadis2006resource,neely2022stochastic,jiao2015virtual,cui2016enhancing,gao2018bias}, 
exhibiting three drawbacks:
1)~\emph{slow startup}: when a flow starts, many packets have to be first backlogged to form stable queue backlog-based gradients, causing large initial end-to-end delay;
2)~\emph{random walk}: during BP scheduling, the fluctuations in queue backlogs drive packets toward random directions, causing unnecessarily long routes or loops;
and 
3)~the \emph{last packet problem} \cite{Alresaini2016bp,ji2012delay}: packets of a short-lived flow could remain enqueued in the network for a long time due to the absence of pressure.
The phenomena of slow startup and random walk in basic BP routing are illustrated by the example in Fig.~\ref{fig:motivation:basic}, in which packets from the red source node did not reach their blue destination in the first 500 time slots, but were trapped in two loops shown by the thickest green edges.

There are two categories of low-overhead BP variations for latency improvement: 
1)~Biased BP adds  pre-defined queue-agnostic biases, e.g., (functions of) shortest path distance \cite{neely2005dynamic,georgiadis2006resource,jiao2015virtual,zhao2023icassp,zhao2023enhanced} to the backlog metric.
Shortest path-biased BP (SP-BP) can mitigate the drawbacks of slow startup and random walk, while maintaining throughput optimality~\cite{neely2005dynamic,georgiadis2006resource,jiao2015virtual} at a low cost, i.e., a one-time communication overhead for bias computation.
As exemplified in Fig.~\ref{fig:motivation:hop}, two major routes around the empty central area can be quickly established by SP-BP \cite{neely2005dynamic,georgiadis2006resource}.
2)~Delay-based BP employ delay-based backlog metrics rather than queue length \cite{ji2012delay,cui2016enhancing,hai2018delay},
in order to address the last packet problem and to stabilize the queues in the network.
Importantly, the aforementioned approaches do not require additional communication overhead at every time step, unlike those based on queue state information (QSI) of the network \cite{cui2016enhancing,gao2018bias} or virtual queues \cite{athanasopoulou2012back,Alresaini2016bp} (see  Section~\ref{sec:review:bp} for 
a more detailed survey).
However, to get the most out of these approaches, careful parameter tuning is required, typically done via trial-and-error.

In this work, we aim to improve the SP-BP routing \cite{neely2005dynamic,georgiadis2006resource,zhao2023icassp} in the presence of short-lived traffic and network mobility, while retaining its low computational and communication overhead and throughput optimality.
Specifically, we seek to improve the widely-used shortest hop-distance bias  \cite{neely2005dynamic,georgiadis2006resource,jiao2015virtual,yin2017improving,ying2010combining} by considering  differences in the scheduling duty cycle of wireless links rather than treating all links equally. 
This can improve routing decisions by avoiding hot-spots with many interfering neighbors (i.e., high betweenness centrality \cite{Katsaros2010social}).
However, link duty cycles are deeply intertwined with routing and scheduling schemes, thus cannot be directly observed by passive monitoring.
As a result, we propose to predict them with a graph neural network (GNN)~\cite{wu2020comprehensive} informed by the interference topology and traffic statistics. 

Furthermore, we also address three longstanding challenges associated with SP-BP~\cite{zhao2023enhanced}: 
1)~how to optimize the scaling of the shortest path bias or edge weights; 
2)~how to efficiently update the shortest path bias under node mobility, e.g., nodes moving, joining or leaving the network; and 
3)~how to effectively incorporate delay-based backlog metrics into SP-BP in response to the last packet problem.


\noindent
{\bf Contribution.} The contributions of this paper are as follows:
\begin{itemize}
    \item \update{We enable the application of GNNs to enhance SP-BP by proposing a model-based framework and a customized offline reinforcement learning (RL) scheme. The GNN incorporates delay awareness into the shortest path biases by predicting link scheduling duty cycles.}
    \item We develop \update{principled approaches for optimal scaling and efficient maintenance of the shortest path biases. Optimal bias scaling mitigates random walk behavior and the last packet problem, while efficient bias maintenance makes SP-BP adaptive to network mobility with low overhead.} 
    \item We propose \update{\emph{expQ}~\cite{zhao2023enhanced}, a new delay-based backlog metric that prioritizes older packets without tracking packet sojourn time. The \emph{expQ} can be seamlessly integrated into SP-BP, and we further prove its throughput optimality.} 
    \item We numerically evaluate \update{our enhancements to SP-BP, demonstrating their effectiveness in improving the latency and packet delivery rate under varying information availability, as well as} their impact on network capacity.

\end{itemize}

\noindent
{\bf Notation:} 
The following notational convention is adopted in this paper: 
$ (\cdot)^\top $, $ \odot $, and $ |\cdot| $ represent the transpose operator,      Hadamard (element-wise) product operator, and the cardinality of a set, respectively.
$ \mathbbm{1}(\cdot) $ is the indicator function.
$ \mathbb{E}(\cdot) $ stands for expectation.
Upright bold lower-case symbol, e.g., $\bbz$, denotes a column vector, and $\bbz_i$ denotes the $i$-th element of vector $\bbz$. 
Upright bold upper-case symbol $\bbZ$ denotes a matrix, of which the element at row $i$ and column $j$ is denoted by $\bbZ_{ij}$, the entire row $i$ by $\bbZ_{i*}$, and the entire column $j$ by $\bbZ_{*j}$.

\section{RELATED WORK}

\subsection{BACKPRESSURE ROUTING}\label{sec:review:bp}
Backpressure (BP) routing and scheduling was first proposed in \cite{tassiulas1992}.
Four categories of BP variations have been developed to improve the latency of basic BP: 
1)~Biased BP adds pre-defined queue-agnostic biases, e.g., (functions of) shortest path distance \cite{neely2005dynamic,georgiadis2006resource,jiao2015virtual,zhao2023icassp,zhao2023enhanced} to the backlog metric.
SP-BP can mitigate the drawbacks of slow startup and random walk, while maintaining throughput optimality~\cite{neely2005dynamic,georgiadis2006resource,jiao2015virtual} at a low cost, i.e., a one-time communication overhead for bias computation.
2)~Delay-based BP replaces queue-length with delay metrics \cite{ji2012delay,cui2016enhancing,hai2018delay} for the backlog.
This approach addresses the last packet problem and maintains throughput optimality. 
3)~Impose restrictions on the routes~\cite{Rai2017loop,yin2017improving} or hop counts~\cite{ying2010combining} to prevent or reduce loops.
However, schemes in 2) and 3) can often shrink the network capacity region.
4)~Use queue-dependent biases that aggregate the queue state information (QSI) of the local neighborhood (or global QSI) to improve myopic BP decisions \cite{cui2016enhancing,gao2018bias} or use shadow queues \cite{athanasopoulou2012back,Alresaini2016bp} to dynamically increase the backpressure.
This, however, requires nodes to collect the neighborhood or global QSI at each time step, leading to high communications costs. 
In addition, most of the aforementioned approaches require careful parameter tuning, typically done via trial-and-error.
Our work seeks to enhance low-overhead SP-BP by incorporating link and topological information, integrating delay-based backlog metrics, and establishing principles for optimal parametric setting.

BP routing has also been extended to broader applications, such as wireless networks with intermittent connections \cite{Ryu2012timescale,Alresaini2016bp} or uncertain network states \cite{Ying2012scheduling}.
More broadly, BP routing has been modified to maximize the throughput of TCP networks~\cite{Liaskos2023} by replacing the backlog-based backpressure with other traffic statistics. 
Another related area is Universal Max Weight~\cite{Sinha2018}, which extends the BP schemes from uni-cast to anycast, multi-cast, and broadcast, and has been adopted in wireless software-defined networks (SDN)~\cite{Nguyen2022}.

\subsection{GNNS FOR WIRELESS RESOURCE ALLOCATION}

GNNs are powerful machine learning (ML) tools that can exploit the communications network structure while maintaining permutation equivariance, a key inductive bias for graph-based tasks, stating that the indexing of nodes in a graph shall not affect their representation~\cite{wu2020comprehensive}. 
As a result, GNNs have been recently applied for various resource allocation problems in wireless networks, including power allocation~\cite{chowdhury2021ml,chowdhury2021unfolding,chowdhury2023unfolding,li2022graph,eisen2020optimal,wang2022stable,shen2023graph,li2023learning} and link scheduling~\cite{zhao2021icassp,zhao2022icassp_a,zhao2022icassp_b,zhao2022twc,Lee2021graph}.
For routing, GNNs have been applied to enhance distributed combinatorial solvers for multicast routing, cluster head election, and virtual backbone establishment in wireless multihop networks~\cite{zhao2023graphbased}, \update{model-free routing algorithm in wired networks~\cite{Li2023gappo},} and to develop digital twins to efficiently predict the performance of routing schemes, replacing computationally costly network simulators~\cite{Rusek2020RouteNet,li2023learnable}.
GNNs have also been applied to inform decisions in computational offloading by processing the subtask dependency relationship and network context~\cite{Zhang2022fine,zhao2023congestionaware}.
To the best of our knowledge, this is the first work to apply GNNs to enhance BP routing and scheduling. 

While supervised learning has been utilized to train GNNs in some tasks like the digital twin~\cite{Rusek2020RouteNet,li2023learnable} for routing performance prediction, generally speaking, generating high-quality training labels can be computationally intractable since many resource allocation tasks are NP-hard~\cite{Lee2021graph}.
To address this challenge, schemes of unsupervised learning~\cite{chowdhury2021ml,chowdhury2021unfolding,chowdhury2023unfolding,li2022graph}, primal-dual learning~\cite{eisen2020optimal,wang2022stable,shen2023graph,li2023learning}, and reinforcement learning (RL)~\cite{zhao2021icassp,zhao2022icassp_a,zhao2022icassp_b,zhao2022twc,zhao2023graphbased,Li2023gappo} have been developed.
Following the methodology in~\cite{zhao2021icassp,zhao2022icassp_a,zhao2022icassp_b,zhao2022twc,zhao2023graphbased} of using GNNs to augment a conventional discrete algorithmic framework, we enhance the biased SP-BP scheme by link duty cycle predicted by GNNs. 
Moreover, our usage of the conflict graph in predicting link duty cycles is inspired by its use in link scheduling in~\cite{zhao2021icassp,zhao2022icassp_a,zhao2022icassp_b,zhao2022twc}. 
Our customized RL scheme enables GNNs to learn by interacting with the networking environment.
Even though our loss function is the mean-square-error between the predicted and empirical link duty cycles, our training approach differs from supervised learning since the predicted link duty cycles will affect  routing and scheduling decisions, and subsequently, the empirical link duty cycles.
However, unlike the RL schemes in~\cite{zhao2021icassp,zhao2022icassp_a,zhao2022icassp_b,zhao2022twc,zhao2023graphbased}, our loss function is decoupled from the system objective, i.e., the average end-to-end delay of data packets.
Thus, our RL scheme circumvents the credit assignment challenge in attributing the contribution of individual link duty cycles to the average end-to-end delay over a chain of SP-BP steps.

\begin{figure}[t!]
	\centering
	\includegraphics[width=0.99\linewidth]{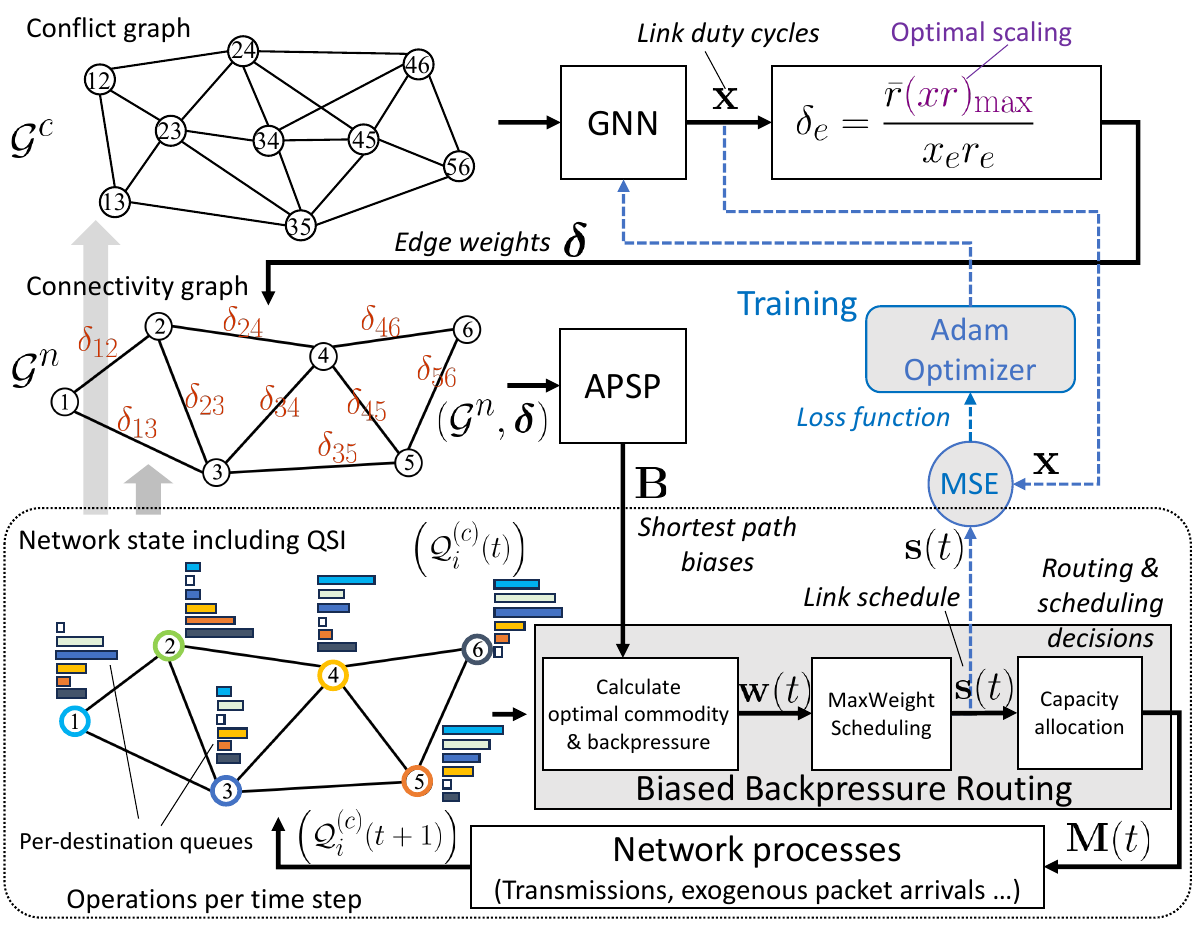}
	\caption{System diagram of our GNN-enhanced SP-BP scheme. }
	\vspace{-0.2in}
	\label{fig:system}
\end{figure}

\section{SYSTEM MODEL AND BACKPRESSURE ALGORITHM}
\label{sec:basic}

{To understand the system model and the basic SP-BP algorithm, the readers can refer to the overall architecture of our proposed SP-BP scheme as illustrated in Fig.~\ref{fig:system}.}

\subsection{SYSTEM MODEL}
\label{sec:sys}

We model a wireless multi-hop network as an undirected graph $\ccalG^{n}=(\ccalV, \ccalE)$, where $\ccalV$ is a set of nodes representing user devices in the network, and $\ccalE$ represents a set of links, where $e=(i,j)\in\ccalE$ for $i,j\in\ccalV$ represents that node $i$ and node $j$ can directly communicate.
$\ccalG^{n}$ is called a connectivity graph and assumed to be a connected graph, i.e., two arbitrary nodes in the network can always reach each other.
Notice that routing involves directed links, so we use ($\overrightarrow{i,j}$) to denote data packets being transmitted from node $i$ to node $j$ over link $(i,j)$.
There is a set of flows $\ccalF$ in the network, in which a flow $f=(i,c)\in\ccalF$, where $i\neq c$ and $i,c\in\ccalV$, describes the stream of packets from a source node $i$ to a destination node (or commodity) $c$, potentially through multiple links.
The medium access control (MAC) of the wireless network is assumed to be time-slotted orthogonal multiple access. 
Each time slot $t$ contains a stage of decision-making for routing and scheduling, followed by a second stage of data transmission.
Each node hosts multiple queues, one for each flow destined to commodity $c\in\ccalV$.
We use {$\ccalQ_{i}^{(c)}(t)$} to denote the QSI of commodity $c$ at node $i$ at the beginning of time slot $t$, and $Q_{i}^{(c)}(t)$ for its queue length.
{Like earlier works on BP~\cite{tassiulas1992,neely2005dynamic,georgiadis2006resource,neely2022stochastic,jiao2015virtual,cui2016enhancing,gao2018bias,Alresaini2016bp,ji2012delay,athanasopoulou2012back,hai2018delay,Rai2017loop,yin2017improving,ying2010combining,Ying2012scheduling,Ryu2012timescale}, we assume an infinite 
queue buffer length to study the impact of the proposed SP-BP on network congestion.}

Matrix $\bbR\in\mathbb{Z}_{+}^{|\ccalE|\times T}$ collects the (stochastic) real-time link rates, of which an element $\bbR_{e,t}$ represents the number of packets that can be delivered over link $e$ in time slot $t$.
The long term link rate of a link $e\in\ccalE$ is denoted by  $r_{e}=\mathbb{E}_{t\leq T}\left[\bbR_{e,t}\right]$, \update{vector of long term link rates $\bbr=\left[r_e|e\in\ccalE\right]$,}
and $\bar{r}=\mathbb{E}_{e\in\ccalE,t\leq T}\left[\bbR_{e,t}\right]$ is the network-wide average link rate.

To describe the conflict relationships between wireless links under orthogonal multiple access, we define \emph{conflict graph}, $\ccalG^c=(\ccalE,\ccalC)$, as follows: a vertex $e\in\ccalE$ represents a link in the original network, and the presence of an undirected edge $(e_1, e_2)\in\ccalC$ means that simultaneous communications on links $e_1, e_2\in\ccalE$ cause interference.
There are two popular models for the conflict relationship between two links: 
1)~Interface conflict, which describes that two links sharing the same node cannot be turned on simultaneously when each node is equipped with only one radio transceiver.
The interface conflict model captures scenarios in which all mobile devices are equipped with antennas with beamforming capabilities, such as mmWave/THz massive MIMO.
The conflict graph under interface conflict model is given by the line graph of the connectivity graph.
2)~Physical distance interference model~\cite{cheng2009complexity}, which arises when two links interfere with each other if their incident nodes are within a certain distance such that their simultaneous transmission will cause the outage probability to exceed a prescribed level.
We consider the unit-disk interference model that captures a simplified scenario in which all the mobile devices transmit at identical power levels with an omnidirectional antenna, and two links conflict with each other if they share the same node or if any of their nodes are within a pre-defined distance.
For the rest of this paper, we assume the conflict graph $\mathcal{G}^{c}$ to be known, e.g., by each link monitoring the wireless channel \cite{zhao2022twc}.
{See the lower left of Fig.~\ref{fig:system} for the QSI, connectivity graph $\ccalG^n$, and conflict graph $\ccalG^c$ of a small exemplary network.}

\subsection{BACKPRESSURE ALGORITHM}
\label{sec:bp}

BP routing and scheduling consist of 4 steps, as illustrated in the block at the lower left of Fig.~\ref{fig:system}.
In step 1, the optimal commodity $c_{ij}^{*}(t)$ on each {directed} link ($\overrightarrow{i,j}$) is selected as the one with the maximal backpressure, 
\begin{equation}\label{E:commodity}
    c_{ij}^{*}(t) =\argmax_{c\in\ccalV}\{ U_{i}^{(c)}(t) - U_{j}^{(c)}(t) \} \;,
\end{equation}
where $U_{i}^{(c)}(t)$ is the backlog metric, whose relationship with queue lengths is detailed in Section~\ref{sec:backlog}. 
For shorter notation, we define the backpressure of commodity $c$ on the directed link ($\overrightarrow{i,j}$) as  $U_{ij}^{(c)}(t) = U_{i}^{(c)}(t) - U_{j}^{(c)}(t)$.
In step 2, the maximum backpressure of ($\overrightarrow{i,j}$) is found as:
\begin{equation}\label{E:weight}
    w_{ij}(t) \!=\!\max\!\left\{ U_{i}^{(c_{ij}^*(t))}\!(t) \!-\! U_{j}^{(c_{ij}^*(t))}\!(t), 0 \right\}.
\end{equation}
In step 3, MaxWeight scheduling \cite{tassiulas1992} finds the schedule $\bbs (t)\in\{0,1\}^{|\ccalE|}$ to activate a set of \emph{non-conflicting links} achieving the maximum total utility, 
\begin{equation}\label{E:scheduling}
    \bbs (t) = \argmax_{\tilde{\bbs} (t)\in \ccalS } ~ \tilde{\bbs}(t)^\top  \left[\bbR_{*,t}\odot\tilde{\bbw}(t)\right] \;,
\end{equation}
where vector $\bbR_{*,t}$ collects the real-time link rate of all links, vector $\tilde{\bbw}(t)=\left[ \tilde{w}_{ij}(t) | (i,j)\in\ccalE\right]$,  and $\tilde{w}_{ij}=\max\{w_{ij}(t),w_{ji}(t)\}\update{\cdot\mathbbm{1}\!\!\left[\!Q_i^{(c_{ij}^*(t))}\!(t)\!>\!0\!\right]}$, and the direction of the link selected by the $\max$ function will be recorded for step 4.
In \eqref{E:scheduling}, $\ccalS$ denotes the set of all non-conflicting configurations, and 
 the per-link utility is  $u_{ij}(t)=\bbR_{ij,t}\tilde{w}_{ij}(t)$.
MaxWeight scheduling involves solving an NP-hard maximum weighted independent set (MWIS) problem \cite{joo2010complexity} on the conflict graph to find a set of non-conflicting links. 
In practice, \eqref{E:scheduling} can be solved approximately by distributed heuristics, such as local greedy scheduler (LGS)~\cite{joo2012local} and its GCN-based enhancement~\cite{zhao2022twc}.
In step 4,  all of the real-time link rate $\bbR_{ij,t}$ of a scheduled link is allocated to its optimal commodity $c_{ij}^{*}(t)$.
The final transmission and routing variables of commodity $c\in\ccalV$ on link ($\overrightarrow{i,j}$) is
\begin{equation}\label{E:quota}
    \mu_{ij}^{(c)}(t) \!=\! \begin{cases}
         \bbR_{ij,t}, & \text{if } c=c_{ij}^{*}(t), w_{ij}(t)>0, s_{ij}(t)=1, \\
         0, & \text{otherwise}.
    \end{cases}
\end{equation}

\subsection{(BIASED) BACKLOG METRICS}
\label{sec:backlog}

The queue length of commodity $c\in\ccalV$ on node $i\in\ccalV$ evolves as follows:
\begin{equation}\label{E:queue}
    Q_i^{(c)}\!(t\!+\!1\!) \!=\! \max\!\left[Q_i^{(c)}\!(t) \!-\! M_{i-}^{(c)}\!(t), 0\right] \! + \! M_{i+}^{(c)}\!(t) \!+\! A_i^{(c)}(t),
\end{equation}
where $M_{i-}^{(c)}(t)=\sum_{j\in\ccalV}\mu_{ij}^{(c)}(t)$ and $M_{i+}^{(c)}(t)=\sum_{j\in\ccalV}\mu_{ji}^{(c)}(t)$  are the numbers of packets of commodity $c$ respectively transmitted and received by node $i$ at time $t$;
$A_i^{(c)}(t)$ is the number of external packets of commodity $c$ injected into node $i$ at time $t$. We assume that the injection follows a random process with an arrival rate $\lambda_i^{(c)}=\mathbb{E}[A_i^{(c)}(t)]$ at source nodes, otherwise it is zero.

The general form of the backlog metric, i.e., $U_{i}^{(c)}(t)$ in \eqref{E:commodity}, in low-overhead BP schemes can be expressed as 
\begin{equation}\label{E:lowcost}
    {U}_{i}^{(c)}(t) = g\left(\ccalQ_{i}^{(c)}(t)\right) + B_{i}^{(c)}  \;,
\end{equation}
where $g(\cdot)$ is a function of QSI, and $ 0\leq B_{i}^{(c)}<\infty $ is a queue-agnostic bias.
In queue length-based BP, $g\bigl(\ccalQ_{i}^{(c)}(t)\bigr)=Q_{i}^{(c)}(t)$, whereas in delay-based BP, $g\bigl(\ccalQ_{i}^{(c)}(t)\bigr)$ can be the sojourn time of the head-of-line (HOL) packet  \cite{ji2012delay} or the entire queue, i.e., sojourn time backlog (SJB) \cite{hai2018delay}.

In unbiased BP, including the basic and delay-based BP, $B_{i}^{(c)}=0$.
In biased BP \cite{neely2005dynamic,georgiadis2006resource,jiao2015virtual,zhao2023icassp}, $B_{i}^{(c)}\geq0$ is defined based on the shortest path distance between node $i$ and destination $c$.
The matrix of all biases is denoted as $\bbB=\bigl(B_{i}^{(c)}\bigr)$ for all $i,c\in\ccalV$.
The shortest path distances can be computed by distributed algorithms for single source shortest path (SSSP) or all pairs shortest path (APSP) on a given network topology, {as illustrated in the middle of Fig.~\ref{fig:system}.}

\section{LINK DUTY CYCLE PREDICTION WITH GNNS}
\label{sec:gnn}

Existing SP-BP algorithms are based on the hop distance (minimal hop count) between two nodes in the networks, which implicitly assumes that every link contributes equally to the end-to-end delay of routing.
This, however, ignores the fact that the average delay of each link is different due to its network locality and link scheduling policy.
In MaxWeight scheduling, the likelihood of a link being scheduled depends on its local conflict or interference topology and the network traffic load.
We propose using a \update{graph convolutional neural network (GCNN)} to predict the expected delay on each link, which can serve to design a  shortest path bias between node $i$ and the destination node $c$ in~\eqref{E:lowcost} better than the shortest hop distance.
{The GNN and its training setup are illustrated in the upper right blocks of Fig.~\ref{fig:system}, while the optimal scaling of link weights is discussed in Section~\ref{sec:solution:scaling}.}
\update{Since our focus is on model-based framework that enables the applications of GNNs for enhancing SP-BP,   
we select GCNN for its simplicity and good generalizability, leaving the search of optimal GNN architectures and hyperparameters for the future work. 
}

\subsection{GRAPH CONVOLUTIONAL NEURAL NETWORKS}
\label{sec:gnn:gnn}

We propose to predict the link duty cycle $\bbx\in\reals^{|\ccalE|}$,  as $\bbx = \Psi_{\ccalG^c}(\boldsymbol{1};\mathbf{\bbomega})$, where  $\Psi_{\ccalG^c}$ is an $L$-layered featureless \update{GCNN} defined on the conflict graph $\ccalG^c$, with a collection of trainable parameters $\bbomega$.
We define the output of an intermediate $l$-th layer of the GCNN as $\bbX^l \in\reals^{|\ccalE|\times g_{l}}$, $\bbX^0 = \boldsymbol{1}^{|\ccalE|\times 1}$, $\bbx = \bbX^L_{*1} $ (the first column of $\bbX^L$),
and the $l$-th layer of the GCNN is expressed as
\begin{equation}\label{E:gcn}
	\mathbf{X}^{l} = \sigma_l\left(\mathbf{X}^{l-1}{\bbTheta}_{0}^{l}+\bbcalL \mathbf{X}^{l-1}{\bbTheta}_{1}^{l}\right), \; l\in\{1,\dots,L\}.
\end{equation}
In~\eqref{E:gcn},  
$\bbcalL$ is the normalized Laplacian of $\ccalG^c$, ${\bbTheta}_{0}^{l}, {\bbTheta}_{1}^{l} \in \mathbb{R}^{g_{l-1} \times g_{l}}$ are trainable parameters (collected in $\bbomega$), and $\sigma_l(\cdot)$ is the activation function of the $l$-th layer. 
The input and output dimensions are set as $g_{0}=1$ and $g_{L}=2$.
The activation functions of the input and hidden layers are selected as leaky ReLUs, whereas a node-wise softmax activation is applied at the output layer (each row of $\bbX^L$). 
{This design limits the scale of output $\bbx$ regardless of the density of input graph $\ccalG^c$.}
Since $\bbcalL$ in \eqref{E:gcn} is a local operator on $\ccalG^c$, each row of $\bbX^{l}$, e.g., $\bbX_{e*}^{l}, e\in\ccalE$ can be computed through neighborhood aggregation as the following local operation on link $e$, 
\begin{equation}\label{E:gcn:local}
    \bbX_{e*}^{l} \!=\! \sigma_l \! \left(\!\bbX_{e*}^{l-1} \, \bbTheta_{0}^{l} \!+\! \left[\! \bbX_{e*}^{l-1} \! - \!\!\!\!\!\sum_{u \in \mathcal{N}_{\ccalG^c}(e)}\!\!\!\frac{\bbX_{u*}^{l-1}}{\sqrt{d({e})d({u})}} \!\right]\!\bbTheta_{1}^{l}\! \right)\!,
\end{equation}
where $\bbX_{e*}^{l}\in\reals^{1\times g_{l}}$ captures the $l$th-layer features on $e$, $\mathcal{N}_{\ccalG^c}(e)$ denotes the set of (interfering) neighbors of $e$, and $d(\cdot)$ is the degree of a vertex in $\ccalG^c$.
Based on~\eqref{E:gcn:local}, the link duty cycle vector $\bbx$ can be computed in a fully distributed manner through $L$ rounds of local message exchanges between $e\in\ccalE$ and its neighbors, making our delay-enhanced BP routing a distributed algorithm.

\subsection{DELAY-AWARE EDGE WEIGHTS}
\label{sec:gnn:delay}

Depending on the availability of the long-term link rate $r_{e}$ and link duty cycle $x_e$ for any link $e\in\ccalE$, we consider the weight of each link for finding the shortest path distance between any pair of nodes to be defined in three possible ways:
$\delta_{e}=1/x_{e}$, $\delta_{e}=1/r_{e}$, or $\delta_{e}=\bar{r}/(x_{e}r_{e})$, where \update{$\bar{r}$ is the network-wide average link rate defined in Section~\ref{sec:sys}}.
\update{In addition,} these link weights can be optimally scaled by $\bar{r}/\delta_{\min}$ based on the principle in Section~\ref{sec:solution:scaling}.
\update{The computation of vector $\bbdelta=\left[\delta_{e}|e\in\ccalE\right]$ is denoted as a function $\bbdelta=f_w(\bbx, \bbr, \bar{r})$.}
By setting the edge weights of the connectivity graph $\ccalG^{n}$ as \update{$\bbdelta$}, bias $B_{i}^{(c)}$ is set as the \emph{weighted} shortest path distance between nodes $i$ and $c$ on $\ccalG^{n}$. 
This distance can be computed with distributed algorithms, such as weighted APSP in general or weighted SSSP when a new node joins the network.

\subsection{COMPLEXITY}
\label{sec:gnn:complexity}

For distributed implementation, the local communication complexity (defined as the rounds of local exchanges between a node and its neighbors) of the GNN is $\ccalO(L)$.
On weighted graphs, the distributed SSSP with the Bellman-Ford algorithm \cite{bellman1958routing,ford1956network} and state-of-the-art APSP~\cite{bernstein2019distributed} both take $\ccalO(|\ccalV|)$ rounds. 
On unweighted graphs, the distributed SSSP and APSP take $\ccalO(D)$ and $\ccalO(|\ccalV|)$  rounds, respectively, where $D$ is the diameter of the unweighted $\ccalG^{n}$.
Compared to the hop distance-based methods~\cite{neely2005dynamic,georgiadis2006resource}, our approach bears additional $\ccalO(L)$ rounds of communications (and larger message size). 
Notice that {the bias matrix} $\bbB$ can be reused over time slots until the network topology ($\ccalG^n$ or $\ccalG^c$) changes, 
which is critical for overhead reduction and scalability promotion.

\begin{algorithm}[!t]
\caption{\update{Pseudo-code for a training or testing instance}}
\label{algo:train}
\update{
\hspace*{\algorithmicindent} \textbf{Input:} $\ccalG^{n}$, $\ccalG^{c}, \ccalF,\bbA,\bbR, T; \bbomega $, training \\
\hspace*{\algorithmicindent} \textbf{Output:} performance metrics, $\bbomega$ (for training only)\\ \vspace{-4mm}
\begin{algorithmic}[1] 
\STATE $r_{e}=\mathbb{E}_{t\leq T}\left[\bbR_{e,t}\right]\; \forall\; e\in\ccalE$; $\bar{r}=\mathbb{E}_{e\in\ccalE,t\leq T}\left[\bbR_{e,t}\right] $
\STATE $ \bbx = \Psi_{\ccalG^c}(\boldsymbol{1};\mathbf{\bbomega}), \bbdelta = f_w(\bbx,\bbr,\bar{r}) $
\STATE Compute $\bbB$ from ($\ccalG^{n}, \bbdelta$) using APSP
\STATE $\ccalQ_{i}^{(c)}(0) = \phi, \forall\; i,c\in\ccalV $ ($\phi$: empty set)
\STATE $\ccalM =\phi$; $t=0$
\WHILE{$t < T$}
\FORALL{ $i\in \ccalV'$ }
\STATE Compute $U_{i}^{(c)}(t)$ based on \eqref{E:lowcost}
\STATE Compute $c_{ij}^*(t)\; \forall\; j\in\ccalN(i)$ based on \eqref{E:commodity}
\STATE Compute $w_{ij}(t)\; \forall\; j\in\ccalN(i)$ based on \eqref{E:weight}
\STATE Compute $\bbs(t)$ with MaxWeight scheduling in \eqref{E:scheduling}
\STATE Compute $\{\mu_{ij}^{(c)}(t) | i,j,c\in\ccalV \}$ based on \eqref{E:quota}
\ENDFOR
\STATE Queueing state evolution: update $ \{\ccalQ_{i}^{(c)}(t+1)|i,c\in\ccalV\} $ based on $ \{\ccalQ_{i}^{(c)}(t)\} $, $\bbA_{*t}, \bbR_{*t}$, and $\{\mu_{ij}^{(c)}(t) \}$
\IF{training}
\STATE $\ccalM \leftarrow \ccalM \cup \{\nabla_{\bbomega}\ell(\bbomega)\} $
\ENDIF
\ENDWHILE
\STATE Collect performance metrics
\IF{training}
\STATE Update $\bbomega$ based on gradient buffer $\ccalM$ with optimizer
\ENDIF
\end{algorithmic}
}
\end{algorithm}

\subsection{TRAINING}
\label{sec:gnn:train}

The parameters $\bbomega$ (collecting ${\bbTheta}_{0}^{l}$  and ${\bbTheta}_{1}^{l}$ across all layers $l$) of our GCNN are trained on a set of routing instances defined on random network processes drawn from a target distribution $\Omega$. 
More precisely, we draw several instances (indexed by $k$) of the network topology, flows, packet arrivals, and link rates $(\ccalG^{n}(k)$, $\ccalG^{c}(k), \ccalF(k),\bbA(k),\bbR(k)) \sim \Omega$. 
For every instance, \update{the training or testing procedure is detailed in Algorithm~\ref{algo:train} and explained as follows:} the GCNN first predicts the link duty cycle vector $\bbx(k)=\Psi_{\ccalG^c(k)}(\boldsymbol{1};\bbomega)$, then biases $\bbB(k)$ are generated by the APSP algorithm based on the link distance vector $\bbdelta(k)=\diag^{-1}(\bbx(k))\boldsymbol{1}$,
then we run the bias-based backpressure routing for $T$ time slots, and collect the schedules for each time slot, $\bbs^{k}(t)$, with the employed LGS scheduler~\cite{joo2012local}.
{As indicated by the blue modules in Fig.~\ref{fig:system},} we train the parameters $\bbomega$ of GCNN to minimize the following mean squared error loss 
\begin{equation}\label{E:sgd_loss}
\ell(\bbomega) \!=\! \mathbb{E}_\Omega \!\left\{\! \frac{1}{|\ccalE|}\! \left\| \bbX^{L}(k) \!-\! \frac{1}{T}\sum_{t=1}^{T}\! \left[\bbs^{k}(t),\boldsymbol{1} \!-\! \bbs^{k}(t)\right] \right\|_{2}^{2} \right\}.
\end{equation}
Intuitively, by minimizing the loss in~\eqref{E:sgd_loss}, we are choosing parameters $\bbomega$ such that the softmax output of our GCNN $\bbX^{L}$ is close to predicting the fraction of time that each link is scheduled.
We do not seek to learn this complex function for \emph{any specific} topology, but rather, we want to minimize the average error over instances drawn from $\Omega$.
In practice, we approximate the expected values in~\eqref{E:sgd_loss} with the corresponding empirical averages.
Indeed, with the collected experience tuples, we update the parameters $\bbomega$ of the GCNN after each training instance, through batch training with random memory sampling, employing the Adam optimizer. 
Notice that, once trained, the GCNN can be used to compute the delay-aware bias for previously unseen topologies without any retraining.
As long as the new topology, arrival rates, and link rates are similar to those observed in $\Omega$ during training, we illustrate in {Section~\ref{sec:results}} that the trained GCNN generalizes to new test instances not observed during training.

\section{IMPROVEMENTS TO BIASED BACKPRESSURE}
\label{sec:solution}

To further improve the latency and practicality of SP-BP schemes, we propose 
three additional enhancements.

\begin{figure}[t]
	\centering
	\subfloat[Time $t$]{
		\includegraphics[height=1.8in]{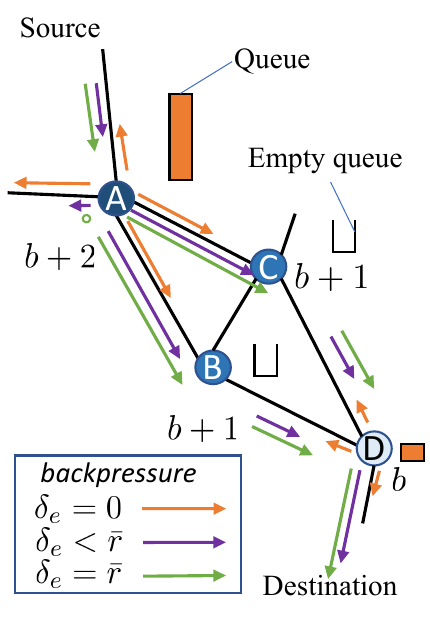}
		\label{fig:example:t0}\vspace{-0.1in}
	}
	\subfloat[Time $t+1$]{
		\includegraphics[height=1.8in]{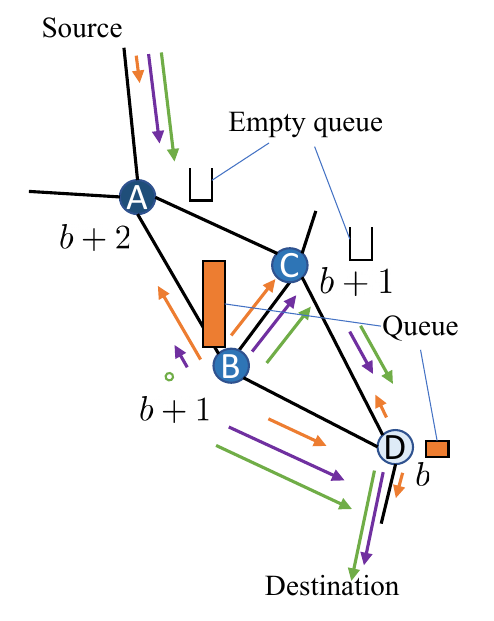}
		\label{fig:example:t1}\vspace{-0.1in}
	}
	\caption{{\small The queue states at four nodes (with hop distance to the destination marked) in an exemplary case of backpressure routing with a single commodity, at time (a) $t$ and (b) $t+1$. An arrow indicates the magnitude and direction of pressure on a link, and its color encodes the choice of the edge weight. All links have an identical $\bar{r}$ rate.
		}
	} 
	\label{fig:example}    
\end{figure}

\subsection{OPTIMAL SCALING OF EDGE WEIGHTS}
\label{sec:solution:scaling}

To the best of our knowledge, the constant scale of the edge weights $\delta_e, e\in \ccalE$ for computing the shortest path biases in SP-BP schemes are often selected via trail-and-error, such as a constant $\delta_{e}$ for all $e\in \ccalE$ in EDR \cite{neely2005dynamic} and $\delta_{e} = \bar{r}/(x_{e}r_{e})$ proposed in Section~\ref{sec:gnn:delay}.
Our goal is to provide a principled method of optimization for such scaling parameter(s).

For biased BP with low-to-medium traffic loads, including the last packets situation, a key observation is that bias scaling influences the random walk behaviors of packets and subsequently the end-to-end delay.
Consider an exemplary situation of last packets in BP routing with only one commodity (flow), for which the queueing states of four nodes in two consecutive time steps are illustrated in Fig.~\ref{fig:example}. 
The hop distances of nodes $A$, $B$, $C$, $D$ to the destination are $b+2$, $b+1$, $b+1$, $b$, respectively. 
Every link has a constant link rate of $\bar{r}$.
At time $t$, $Q_{A}^{(c)}(t)=\bar{r}$ and $Q_{B}^{(c)}(t)=Q_{C}^{(c)}(t)=0$, the backpressures point from $A$ towards $B$ and $C$. 
Assume that at the beginning of time slot $t+1$, all the packets on $A$ moved to $B$.
Different choices of $\delta_{e}$ would lead to different routing decisions at $t+1$.

For unbiased BP with $\delta_{e}=0$, the backpressures (orange arrows) are pointed from $B$ towards $A$, $C$, and $D$, and $U^{(c)}_{BA}(t+1) = U^{(c)}_{BC}(t+1) = \bar{r} > U^{(c)}_{BD}(t+1)$. 
At the end of time slot $t+1$, the packets on $B$ will move to either $A$ or $C$, packets on $D$ may also move back to $C$, causing an unwanted meandering of the packets.
For biased BP with $\delta_{e}>0$, the backpressures (magenta and green arrows) at time $t+1$ originating from $B$ are skewed towards the destination, $U^{(c)}_{BD}(t+1) > U^{(c)}_{BC}(t+1) =\bar{r} > U^{(c)}_{BA}(t+1) $.
At the end of time slot $t+1$, packets in $B$ will move forward to $D$.
In particular, when $\delta_e=\bar{r}$, the backpressure (i.e., green arrows) $U^{(c)}_{BA}(t+1)=0$.

The previous example can be formally described as follows.
Consider a wireless multi-hop network with homogeneous link rates, i.e., all links have a link rate of $\bar{r}$, 
there are two links $(i,j)$ and $(j,k)$, where nodes $j, k$ are on the shortest path from node $i$ to node $c$.
Further, consider a congestion-free last packets (CFLP) scenario: 
at time $t$, the last packets of commodity $c$ reside on node $i$, and 
link $(i,j)$ is congestion-free, i.e., $ 0< Q_{i}^{(c)}(t)=q\leq \bar{r}, Q_{j}^{(c)}(t)=0 $ and no external packets of commodity $c$ will arrive at nodes $i,j,k$ from the rest of the network or users. 
If link $(i,j)$ is scheduled at time $t$, such that $ Q_{i}^{(c)}(t+1)=0, Q_{j}^{(c)}(t+1)=q\leq \bar{r} $, we have the following Lemma.
\begin{lemma}\label{l:back}
In SP-BP routing under the CFLP scenario, the edge weight $\delta_e$ should be greater than or equal to the homogeneous link rate,  $\delta_e\geq\bar{r}$, to avoid the immediate reversal of the direction of backpressure for commodity $c$ on the scheduled link $(\overrightarrow{i,j})$ after the transmission. 
\end{lemma}
\vspace{-0.1in}
\begin{proof} 
In queue length-based SP-BP, the backpressure on link $(\overrightarrow{i,j})$ is  $U_{ij}^{(c)}\!(t) = Q_{ij}^{(c)}\!(t) + \big(B_{i}^{(c)}\!-\!B_{j}^{(c)}\big) = q+\delta_{e}$. 
Since link $(\overrightarrow{i,j})$ is scheduled at time $t$, we have that $U_{ij}^{(c)}\!(t+1)=\delta_{e}-q$. 
To avoid the direction of the backpressure on link $(i,j)$ being reversed at time $t+1$, i.e., $ U_{ij}^{(c)}\!(t\!+\!1)<0 $, we need to set $\delta_{e}\geq\bar{r}$.  
\end{proof}

In the described CFLP scenario, Lemma~\ref{l:back} ensures that the backpressure algorithm does not oscillate as the last packets travel through the shortest path when $\delta_e\geq\bar{r}$.
On the other hand, a smaller value of $\delta_e$ is preferable for path finding and congestion prevention.
Indeed, an extremely large SP bias scaling $\delta_e$ would force every packet through the shortest path, thus hindering any information from the queue lengths and defeating the original purpose of BP routing.
Combining these two observations, \emph{we advocate for setting the minimal edge weight as the average link rate}, i.e., $\underset{e\in\ccalE}{\min}\;\delta_e\vcentcolon=\bar{r}$. In particular, for EDR in~\cite{neely2005dynamic}, this implies setting $\delta_e =\bar{r}$.

Although the above discussion is based on the CFLP setting with all link rates equal to the average $\bar{r}$, in Section~\ref{sec:results}, we demonstrate experimentally the optimality of the above choice under heterogeneous link rates and general traffic settings.

\subsection{BIAS MAINTENANCE} 
Due to the mobility and distributed nature of networks, bias maintenance requires frequent or periodical SSSP and/or APSP computation.
Although the biases $\bbB$ can be re-used for many time slots to match the slowly changing topology, the overhead of bias maintenance can be high for large networks.

To address this issue, we propose {to adopt the neighborhood update rule from Dijkstra's algorithm~\cite{dijkstra1959note}} for bias maintenance at node $i\in\ccalV$, whenever one or more of its incident links are established or destroyed 
\begin{equation}\label{E:neighbor}
    B_{i}^{(c)}(t+1) = \begin{cases}
    \underset{j\in\ccalN_{\ccalG^n}(i)}{\min} \left[B_{j}^{(c)}(t) + \delta_{ij}(t) \right], \;& i \neq c \\
    0, & i = c
    \end{cases}\;,
\end{equation}
where $\delta_{ij}(t)$ is the edge weight between neighboring nodes $i$ and $j$.
This rule allows the shortest path bias to be updated within $\ccalO(D)$ steps of local message exchange.

\subsection{DELAY-AWARE BACKLOG METRIC: EXPQ}\label{sec:solution:expQ}
To prioritize old packets, we further introduce a SP-BP-friendly delay-aware QSI function $g(\cdot)$ in \eqref{E:lowcost}: 
\begin{equation}\label{E:expQ}
\begin{aligned}
    g\!\left(\!\ccalQ_{i}^{(c)}\!(t+1)\!\right)\!=&(1\!+\!\epsilon)g\!\left(\!\ccalQ_{i}^{(c)}(t)\!\right)\!\max\!\left[\!1\!-\!\frac{M_{i-}^{(c)}\!(t)}{Q_{i}^{(c)}(t)}, 0\right] \\
    & + M_{i+}^{(c)}\!(t) + A_i^{(c)}(t)\;,
\end{aligned}
\end{equation}
where $\epsilon\geq 0$ is a small constant parameter, and it follows that $ g\big(\ccalQ_{i}^{(c)}(t)\big)=Q_{i}^{(c)}(t) $ for $\epsilon=0$ and that $g\left(\varnothing\right)=0$ for an empty queue $\varnothing$.
With this definition, it is easy to find that $ g\big(\ccalQ_{i}^{(c)}(t)\big) \geq Q_{i}^{(c)}(t)$. 
Further restrictions on $\epsilon$ can be found in \eqref{Ea:eps} and the following discussion in Section~\ref{sec:proof:expQ}.

The $g(\cdot)$ in \eqref{E:expQ} is named \emph{expQ}, as it will increase exponentially over time if commodity $c$ at node $i$ is not scheduled.
Compared to the backpressure metric based on sojourn time backlog \cite{hai2018delay} and HOL sojourn time \cite{ji2012delay}, \emph{expQ} can increase the bandwidth utilization efficiency, since it does not require tracking the sojourn time of individual packets. 

\section{THROUGHPUT OPTIMALITY}\label{sec:proof}

The throughput optimality of a routing scheme states that all queues in the network are strongly stable as long as the arrival rates of flows are within the network capacity region.
The condition of strong stability is defined as\cite{neely2005dynamic,georgiadis2006resource,neely2022stochastic,jiao2015virtual} 
\begin{equation}\label{E:stability}
    \limsup_{t\rightarrow \infty}\frac{1}{t}\sum_{\tau=0}^{t-1}\mathbb{E}\left[Q_i^{(c)}(\tau)\right]<\infty\;,\forall\; i,c\in\ccalV\;.
\end{equation}
Next, we first prove the throughput optimality of queue length-based SP-BP through Lyapunov drift theory~\cite{neely2005dynamic,georgiadis2006resource,neely2022stochastic,jiao2015virtual}, then the conditional throughput optimality for SP-BP with \emph{expQ}.

\subsection{QUEUE LENGTH-BASED SP-BP SCHEME}\label{sec:proof:spbp-q}
Our proof is similar to~\cite{neely2005dynamic,georgiadis2006resource,neely2022stochastic,jiao2015virtual}, except for \eqref{E:bq}-\eqref{E:w}.

\begin{theorem}\label{t:spbp-q}
A SP-BP scheme with backlog metric, $U_i^{(c)}(t)=Q_i^{(c)}(t)+B_i^{(c)}$ \forall $i,c\in\ccalV$, where $0\leq B_i^{(c)}<\infty$ is a non-negative queue-agnostic constant bias, meets the condition of strong stability in \eqref{E:stability} within the network capacity region.
\end{theorem}
\vspace{-0.1in}
\begin{proof} 
Define $\bbQ(t)=\left(\ccalQ_i^{(c)}(t)\right)_{i,c\in\ccalV}$ as the current queueing state of the network. Define the Lyapunov function as
\begin{equation}\label{E:Lyapunov}
    L(t):=\frac{1}{2}\sum_{i\in\ccalV}\sum_{c\in\ccalV}U_i^{(c)}(t)^2\;,
\end{equation}
and the Lyapunov drift as
\begin{equation}
    \Delta(t):=\mathbb{E}\left[ L(t+1) - L(t) | \bbQ(t) \right]\;.
\end{equation}
The evolution of metric $U_i^{(c)}(t)=Q_i^{(c)}(t)+B_i^{(c)}$ follows 
\begin{equation}\label{E:bq}
    U_i^{(c)}\!(t\!+\!1\!) \!=\! \max\!\left[\!U_i^{(c)}\!(t) \!-\! M_{i-}^{(c)}\!(t), B_i^{(c)}\!\right] \! + \! M_{i+}^{(c)}\!(t) \!+\! A_i^{(c)}\!(t).
\end{equation}
By squaring the queue update in \eqref{E:bq}, and using the fact that $(\max[q-o, b]+p)^2\leq q^2+o^2+p^2+b^2+2bp+2q(p-o)$ for $q\geq0, o\geq0, b\geq 0, p\geq 0$ (see Appendix~\ref{app:inequality}), we have:
\begin{equation}\label{E:drift}
    \Delta(t) \!\leq\! W \!\!+\!\!\! \sum_{i,c\in\ccalV}\!\!U_i^{(c)}\!(t)\mathbb{E}\!\left[\!\lambda_i^{(c)} \!+\! M_{i+}^{(c)}\!(t) \!-\! M_{i-}^{(c)}\!(t) \Big| \bbQ\!(t)\!\right]\!,
\end{equation}
where $\mathbb{E}\!\left[ A_{i}^{(c)}\!(t) \Big| \bbQ\!(t)\right]\!=\!\lambda_i^{(c)}$ 
and $W$ is a finite constant that depends on the maximum possible first and second moments of packet arrival rates and transmission rates:
\begin{equation}\label{E:w}
\begin{aligned}
W \geq & \frac{1}{2}\!\sum_{i,c\in\ccalV}\!\mathbb{E}\bigg\{ 
     \left[M_{i-}^{(c)}(t)\right]^2 + \left[M_{i+}^{(c)}\!(t) \!+\! A_i^{(c)}(t)\right]^2  \\
    &  + \left[B_{i}^{(c)}\right]^2 + 2B_{i}^{(c)}\!\left[M_{i+}^{(c)}\!(t) \!+\! \lambda_i^{(c)}\right] \bigg\}\;.
\end{aligned}
\end{equation}
Next, we show that the biased BP described in Section~\ref{sec:bp} is designed to choose a feasible policy $\bbM(t)=\big(\mu_{ij}^{(c)}(t)\big)_{i,j,c\in\ccalV}$  under current network state to minimize the bound of the Lyapunov drift, i.e., the right-hand side of \eqref{E:drift}. 
Since $W$ and $\lambda_{i}^{(c)}$ are constants, this amounts to minimizing the following:
\begin{equation*}
\mathbb{E}\bigg\{\sum_{i,c\in\ccalV}\!U_i^{(c)}\!(t)\left[ M_{i+}^{(c)}(t) \!-\! M_{i-}^{(c)}(t) \right] \bigg| \bbQ(t)\bigg\}\;,
\end{equation*}
where the finite sums in the second term in \eqref{E:drift} are pushed through the expectations. 
By the principle of opportunistically minimizing an expectation, we can minimize the above expectation by minimizing the function inside of it given the observed $\bbQ(t)$ and network topology (network state): 
\begin{equation}\label{E:inner:min}
   \bbM^*\!(t) \! = \!\argmin_{\bbM(t)\in\bbPi} \sum_{i,c\in\ccalV}\!\!U_i^{(c)}\!(t)\!\Bigg[\sum_{j\in\ccalV}\mu_{ji}^{(c)}(t) \!-\! \sum_{j\in\ccalV}\mu_{ij}^{(c)}(t) \!\Bigg]\!,
\end{equation}
where $\bbPi$ is the feasible policy space of routing and scheduling decisions.
By rearranging the sums in \eqref{E:inner:min}, we have   
\begin{equation}\label{E:inner:max}
  \bbM^*(t)  = \argmax_{\bbM(t)\in\bbPi} \sum_{i\in\ccalV}\sum_{j\in\ccalV}\sum_{c\in\ccalV}\!\mu_{ij}^{(c)}\left[ U_{i}^{(c)}(t) \!-\! U_{j}^{(c)}(t) \right].
\end{equation}
Notice that \eqref{E:inner:max} is equivalent to SP-BP in \eqref{E:commodity}-\eqref{E:quota}.
Since the bound on the Lyapunov drift in~\eqref{E:drift} is minimized, the arrival rates  within the network capacity region $\big(\lambda_{i}^{(c)} \big)_{i,c\in\ccalV} \in \bbLambda$ can be expressed as that there is a constant $\beta>0$ such that   
$$
\mathbb{E}\left[\lambda_i^{(c)} \!+\! M_{i+}^{(c)}(t) \!-\! M_{i-}^{(c)}(t) \Big| \bbQ(t)\right] \leq -\beta ,\;\forall\; i,c\in\ccalV.
$$
Substituting the expectation in \eqref{E:drift} with $-\beta$, we have
\begin{equation}\label{E:bound}
    \Delta(t) \leq W - \beta\sum_{i,c\in\ccalV}\!U_i^{(c)}\!(t)\;.
\end{equation}
The time average of \eqref{E:bound} yields \cite[Lemma 4.1]{georgiadis2006resource}:
\begin{equation}\label{E:stable:bias}
    \limsup_{t\rightarrow \infty}\frac{1}{t}\sum_{\tau=0}^{t-1}\sum_{i,c\in\ccalV}\mathbb{E}\left[Q_i^{(c)}(\tau)+B_i^{(c)}\right]<\frac{W}{\beta} \;.
\end{equation}
Since $W \!<\!\infty$, $0\!\leq\!B_i^{(c)}\!<\!\infty$, the condition of strong stability in \eqref{E:stability} immediately follows by moving $\sum_{i,c\in\ccalV}B_i^{(c)}$ from the left hand side of~\eqref{E:stable:bias} to the right hand side.
\end{proof}

\vspace{-0.15in}
\subsection{SP-BP SCHEME WITH EXPQ}\label{sec:proof:expQ}
\begin{theorem}\label{t:expQ}
There exists a constant $ \epsilon \geq 0 $ such that SP-BP with delay-based backlog metric \emph{expQ} defined in Section~\ref{sec:solution:expQ}, denoted as $ E_i^{(c)}\!(t)=g\!\left(\!\ccalQ_{i}^{(c)}(t)\!\right) $, meets the condition of strong stability in~\eqref{E:stability} within a reduced network capacity region. 
\end{theorem}
\begin{proof}
We denote the backlog metric of SP-BP with \emph{expQ} as $ \tilde{U}_i^{(c)}(t)=E_i^{(c)}(t)+B_i^{(c)} $, and consider the compounding component of $E_i^{(c)}(t)$ in \eqref{E:expQ} as virtual packet arrivals, which is assumed to be bounded for a finite $t<\infty$ and some $\epsilon\geq 0$:
$$ \epsilon E_{i}^{(c)}(t)\!\max\!\left[\!1\!-\!\frac{M_{i-}^{(c)}\!(t)}{Q_{i}^{(c)}(t)}, 0\right] \leq \epsilon E_{i}^{(c)}(t) \leq b\;. $$
This, along with \eqref{E:expQ} and $ E_{i}^{(c)}(t) \geq Q_{i}^{(c)}(t)$, lead to
\begin{equation}\label{Ea:bp}
    \tilde{U}_i^{(c)}\!(t\!+\!1\!) \!\leq\! \max\!\left[\!\tilde{U}_i^{(c)}\!(t) \!-\! M_{i-}^{(c)}\!(t), B_i^{(c)}\!\right] \! + \! M_{i+}^{(c)}\!(t) \!+\! A_i^{(c)}\!(t) \!+\! b.
\end{equation}
We apply the proof for Theorem~\ref{t:spbp-q} by substituting \eqref{E:bq} with~\eqref{Ea:bp}, and replacing $ U_i^{(c)}(t) $ by $\tilde{U}_i^{(c)}(t)$ and $ A_{i}^{(c)}\!(t) $ by $ A_{i}^{(c)}\!(t) + b $, thus obtain the policy of SP-BP with \emph{expQ} as
\begin{equation}\label{E:expQ:max}
  \bbM^*(t)  = \argmax_{\bbM(t)\in\bbPi} \sum_{i\in\ccalV}\sum_{j\in\ccalV}\sum_{c\in\ccalV}\!\mu_{ij}^{(c)}\left[ \tilde{U}_{i}^{(c)}(t) \!-\! \tilde{U}_{j}^{(c)}(t) \right],
\end{equation}
and that within a reduced capacity region $ \big(\lambda_{i}^{(c)} + b\big)_{i,c\in\ccalV} \in \bbLambda $, there is a constant $\beta>0$ such that \forall $ i,c\in\ccalV $ 
$$
\mathbb{E}\left[\lambda_i^{(c)} \! + b \!+\! M_{i+}^{(c)}\!(t) \!-\! M_{i-}^{(c)}\!(t) \Big| \bbQ(t)\right] \leq -\beta \; .
$$
Subsequently, \emph{expQ} meets the condition of strong stability  
\begin{equation}\label{E:stable:expQ}
    \limsup_{t\rightarrow \infty}\frac{1}{t}\sum_{\tau=0}^{t-1}\sum_{i,c\in\ccalV}\mathbb{E}\left[E_i^{(c)}(\tau)\right]<\frac{W}{\beta} - \sum_{i,c\in\ccalV}B_i^{(c)}\;.
\end{equation}
With $E_i^{(c)}\!(0)\!<\!\infty$ and $|E_i^{(c)}\!(t\!+\!1)\!-\!E_i^{(c)}\!(t)|\!<\!\infty$, \eqref{E:stable:expQ} leads to uniform boundness $E_i^{(c)}\!(t) \!<\! \gamma \!<\! \infty$. The condition for the above process to be true for $t\rightarrow \infty$ is  
\begin{equation}\label{Ea:eps}
    0\leq\epsilon \leq \frac{b}{\gamma}\;,\; \text{where}\; \left(\!\lambda_{i}^{(c)} \!+\! b \!\right)_{i,c\in\ccalV}\! \in \! \bbLambda,\; E_i^{(c)}\!(t) \!<\! \gamma\;.
\end{equation}
With \eqref{E:stable:expQ}, \eqref{Ea:eps} and ${E}_{i}^{(c)}\!(t) \geq {Q}_{i}^{(c)}\!(t) $, 
all queues are strongly stable within the network capacity region shrank by $b$. 
\end{proof}

According to \eqref{Ea:eps}, in order to come up with a reasonable $\epsilon$ for \emph{expQ}, we need an allowable reduction of capacity region, $b$, and an estimated upper bound of $E_i^{(c)}(t)$.

Simulations in \cite[Ch~4.9]{georgiadis2006resource} and Section~\ref{sec:results:capacity} show that for BP with queue length-based backlog metrics, the inclusion of non-negative queue-agnostic biases does not shrink the network capacity region of the basic BP.
However, as illustrated in Section~\ref{sec:results:capacity}, the application of delay-based backlog metrics, including \emph{expQ}, will reduce the network capacity of BP.

\section{NUMERICAL EXPERIMENTS}
\label{sec:results}

To illustrate the effectiveness of our proposed enhancements to the SP-BP schemes in improving routing performance, we evaluate a series of low-overhead BP routing algorithms in simulated wireless multi-hop networks.
Each network instance is generated by a 2D point process with a given number of nodes $|\ccalV|\in\{20,30,\dots,110\}$ uniformly distributed in the plane with a constant density of $8/\pi$.
We consider a simplified scenario where all mobile devices are equipped with omnidirectional antennas and transmit at identical power levels, which can be captured by the unit-disk interference model previewed in Section~\ref{sec:sys}.
A link is established between two nodes if they are located within a distance of $1$ unit. 
The average degree of conflict graphs under this test configuration is $34.6$.
For each network size $|\ccalV|$, we generate 100 test instances by drawing 10 instances of random networks, each with 10 realizations of random source-destination pairs and random link rates. 
A test instance contains a number (uniformly chosen between $\lfloor 0.30 |\ccalV|\rfloor$ and $ \lceil 0.50|\ccalV|\rceil$) of random flows between different pairs of sources and destinations.
Each test instance also includes a realization of uniformly distributed long-term link rates, $r_{e} \sim \mathbb{U}(10,42)$, which are the expected values of the real-time link rates, $\bbR_{e,t} \sim  \mathbb{N}(r_{e}, 3)$, truncated to between $r_e\pm 9$.
The link rates are configured to capture fading channels with lognormal shadowing.
A test instance is simulated for a total of $T=1000$ time steps.

Two types of flows are employed: in a \textit{streaming} flow $f$, new packets arrive at the source following a Poisson process with a uniformly random arrival rate $\lambda(f) \sim \mathbb{U}(0.2,1.0)$ for all time slots; 
in a \textit{bursty} flow, packets follow Poisson arrivals with a random rate $\lambda(f) \sim \mathbb{U}({2.0,10.0})$ (lighter traffic) or $\lambda(f) \sim \mathbb{U}({6.6,33.0})$ (heavier traffic) for $t<30$, and no new packet arrivals for $t\geq30$.  
With $T=1000$, \update{a streaming and a heavier bursty flows are expected to inject similar  total numbers of packets into the network, allowing better assessment of the impact of traffic pattern}.
\update{These flow rates are considered lightweight since the expected exogenous packet arrival rate per node of $ 0.24 $ is only $8\%$ of the expected total maximum outflow rate per node of $3$, see Appendix~\ref{app:frate} for details.}
The end-to-end delay of a tested BP scheme is collected by tracking the time each packet arrived at the source node $t_0$ and the destination node (departure time) $t_1$.
The end-to-end delay of an undelivered packet is treated as $T$.
We record the latency of individual flows (source-destination pairs) by averaging the end-to-end delays of their packets.

\begin{table}[!t]
    \centering
    \caption{\update{Hyperparameters for Training}}
    \label{tab:training}
\update{
    \begin{tabular}{l|p{5.5cm}}
        \textbf{Item} & \textbf{Value}  \\ \hline
        GCNN config & $L=5$, $g_{l}=32$ for $ 1\leq l\leq 4, g_0=1, g_L=2$, see  Section~\ref{sec:gnn:gnn} for other details  \\ \hline
        Regularization & L2 regularization with a weight decay of $0.001$ \\ \hline
        Learning rate & 0.0001 \\ \hline
        Simulation steps & $T=200$ per training instance \\ \hline
        Mini-batch size & 1 training instance \\ \hline
        Training graphs & A set of 5000 random networks, with size $|\ccalV|\in\{20,30,40,50,60\} $, 100 instances per $|V|$ \\ \hline
        Interference radius & Uniformly random $\mathbb{U}(0, 0.8)$ \\ \hline
        Link rates & Same as test setting \\ \hline
        Number of flows & Uniformly chosen between $\lfloor 0.15 |\ccalV|\rfloor$ and $ \lceil 0.30|\ccalV|\rceil$ \\ \hline
        Traffic pattern & All streaming flows with $\lambda(f) \sim \mathbb{U}(0.2,1.0)$ \\ 
        \hline
    \end{tabular}
}
\end{table}

\begin{table*}[t!]
	\vspace{-0.2in}
	\renewcommand{\arraystretch}{1.0}
	\caption{Tested low-overhead BP algorithms and their detailed configurations, information requirements, and result summary.
	} 
	\label{tab:list}
    \setlength{\tabcolsep}{4pt}
	\footnotesize
	\begin{tabular}{|l|l|l|l|l|l|l|l|l|l|l|}
        \hline
        SP-BP algorithm  & Link & Bias & Backlog & Local & Global & \multicolumn{5}{|c|}{Means$^{*}$ (ranking) under streaming (s) \& bursty (b) traffic} \\ \cline{7-11}
        Acronym & Distance & Scaling & $U_{i}^{(c)}(t)$ & Info.$^\ddag$ & Info. & latency (s) & delivery (s) & latency (b) & delivery (b) & network cap.$^\dag$(s) \\ \hline
        BP & 0  & -- & $Q$  & -- & -- & 891.1 (15) & 12.9\% (15) & 870.1 (15) & 14.7\% (15) & 32.83 (8) \\ \hline
        BP-HOL & 0 & -- & HOL & $t_0$ & -- & 890.3 (14) & 13.9\% (14) & 830.4 (14) & 23.1\% (14) & 15.86 (15) \\ \hline
        BP-SJB & 0 & -- & SJB & $t_0$ & -- & 874.1 (13) & 16.1\% (13) & 792.8 (13) & 26.9\% (13) & 29.89 (13) \\ \hline
        EDR-$\bar{r}$ & 1 & $\bar{r}$ & $Q$ & -- & $\bar{r}$ & 196.7 (10) & 88.1\% (10) & 245.3 (10) & 92.1\% (11) & 49.45 (4) \\ \hline
        EDR-$\bar{r}$-expQ & 1 & $\bar{r}$ & $Q$ & -- & $\bar{r}$ & 194.3 (9) & 89.2\% (8) & 227.0 (7) & 96.6\% (8) & 30.95 (10) \\ \hline
        EDR-$\bar{r}$-HOL & 1 & $\bar{r}$  & HOL & $t_0$ & $\bar{r}$ & 261.0 (11) & 83.5\% (11) & 284.3 (11) & 93.8\% (9) & 27.55 (14) \\ \hline
        EDR-$\bar{r}$-SJB & 1 & $\bar{r}$  & SJB & $t_0$ & $\bar{r}$ & 804.2 (12) & 24.3\% (12) & 713.5 (12) & 38.0\% (12) & 30.04 (12) \\ \hline
        SP-$1/x$ & $1/x$ & -- & $Q$  & $x$ & -- & 187.8 (7) & 88.9\% (9) & 240.5 (9) & 92.8\% (10) & 49.33 (5) \\ \hline
        SP-$1/x$-expQ & $1/x$ & -- & \emph{expQ}  & $x$ & -- & 191.9 (8) & 89.4\% (7) & 227.2 (8) & 96.6\% (7) & 30.90 (10) \\ \hline
        SP-$1/r$-min & $1/r$ & $\bar{r}r_{\max}$ & $Q$  & $r$ & $\bar{r}$, $r_{\max}$ & 90.0 (2) & 96.3\% (2) & 169.9 (2) & 98.7\% (4) & 54.43 (2) \\ \hline
        SP-$1/r$-min-expQ & $1/r$ & $\bar{r}r_{\max}$ & \emph{expQ}  & $r$ & $\bar{r}$, $r_{\max}$ & 115.6 (5) & 95.0\% (5) & 176.4 (4) & 99.1\% (2) & 37.71 (7) \\ \hline
        SP-$\bar{r}/(xr)$ & $\bar{r}/(xr)$ & -- & $Q$  & $r$, $x$ & $\bar{r}$ & 103.0 (3) & 95.6\% (3) & 191.2 (5) & 97.9\% (6) & 49.38 (3) \\ \hline
        SP-$\bar{r}/(xr)$-expQ & $\bar{r}/(xr)$ & -- & \emph{expQ}  & $r$, $x$ & $\bar{r}$ & 143.6 (6) & 93.0\% (6) & 207.6 (6) & 98.1\% (5) & 31.95 (9) \\ \hline
        SP-$\bar{r}/(xr)$-min & $\bar{r}/(xr)$ & $(xr)_{\max}$ & $Q$  & $r$, $x$ & $\bar{r}$, $(xr)_{\max}$ & 88.7 (\textbf{1}) & 96.4\% (\textbf{1}) & 166.9 (\textbf{1}) & 98.8\% (3) & 55.26 (\textbf{1}) \\ \hline
        SP-$\bar{r}/(xr)$-min-expQ & $\bar{r}/(xr)$ & $(xr)_{\max}$ & \emph{expQ}  & $r$, $x$ & $\bar{r}$, $(xr)_{\max}$ & 111.3 (4) & 95.2\% (4) & 171.4 (3) & 99.2\% (\textbf{1}) & 38.82 ({6}) \\ \hline
	\end{tabular}
 \newline\newline  
 \raggedright
 {$t_0$: the time a packet arrives to the network, $r$: individual long-term link rate, $\bar{r}$: network average long-term link rate, $x$: predicted link duty cycle,  $Q$: queue length, SJB: total sojourn time of all packets in a queue, HOL: sojourn time of head-of-line packet in a queue, \emph{expQ}: backlog func. defined in \eqref{E:expQ}. \newline 
 $\ddag$ All BP algorithms require instantaneous link rates $\bbR_{e,t}$ and conflicting degree $d(e)$ for $e\in\ccalE$ in MaxWeight scheduling. Computing $x_e$ requires $d(e)$. \newline 
 $*$ Latency and delivery rate based on mixed traffic; units: latency (number of time slots), network capacity (number of packets per time slot).\newline 
 $\dag$ Ranked by maximum network throughput on random networks of 100 nodes (Fig.~\ref{fig:regular:load}); when maximal values are close, values for $\lambda\leq3.0$ weigh in.}
 \vspace{-0.05in}
\end{table*}

We train a 5-layer GCNN ($L=5$, $g_{l}=32,l\in\{1,\dots,4\}$) as detailed in Section~\ref{sec:gnn:gnn} according to the RL scheme in Section~\ref{sec:gnn:train}.
The training instances are generated similar to the test instances described earlier, except that the underlying network instances are from a set of 100 random networks with $|\ccalV|\in\{20,30,\dots,60\}$, and the number of flows is uniformly chosen between $\lfloor 0.15 |\ccalV|\rfloor$ and $ \lceil 0.30|\ccalV|\rceil$.
Notice that our test instances cover larger networks and denser traffic than our training configuration,
to illustrate that our GNN-based SP-BP schemes can generalize to problems of larger scales.
\update{The detailed hyperparameters and training setting are listed in Table~\ref{tab:training},
and the model converged after 5 epochs.\footnote{Training takes 5 hours on a workstation with a specification of 32GB memory, 8 cores, and Geforce GTX 1070 GPU. }}

\begin{figure}
\centering
    \includegraphics[width=0.95\linewidth]{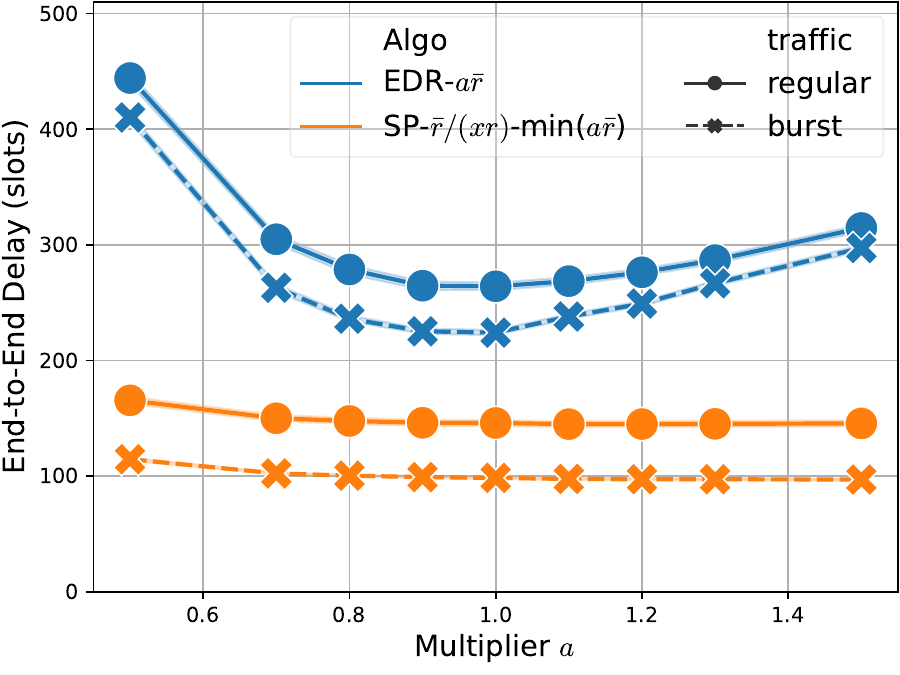}
    \vspace{-0.1in}
    \caption{{\small 
    End-to-end delay of SP-BP routing algorithms as a function of the multiplier of minimum edge weight  $\min(\delta_{e})=a\bar{r}$ on networks of $100$ nodes, simulated with unit-disk interference model, long-term link rates $r_{e} \sim \mathbb{U}(10,42)$, total time steps $T=1000$, and $100$ test instances per point.
    Flow rate $\lambda(f) \sim \mathbb{U}(0.2,1.0)$ for streaming traffic,
    $\lambda(f) \sim \mathbb{U}(2.0,10.0)$ for bursty traffic when $t<30$. 
    }}
    \label{fig:regular:mpx}
    \vspace{-0.1in}
\end{figure}

The tested BP algorithms (acronyms in parentheses identify legends in figures) include 1) unbiased BP: basic BP (BP), delay-based unbiased  BP with SJB \cite{hai2018delay} (BP-SJB) and HOL \cite{ji2012delay} (BP-HOL); 2) queue-length based biased BP: EDR \cite{neely2005dynamic} with edge weight $\delta_{e}=\bar{r}$ (EDR-$\bar{r}$), and GCN-based delay-aware shortest path bias \cite{zhao2023icassp} (SP-$\bar{r}/(xr)$) with edge weight $\delta_{e}=\bar{r}/(x_{e}r_{e})$, and its scaled version (SP-$\bar{r}/(xr)$-min); and
3) delay-based biased BP: combination of BP-SJB and EDR-$\bar{r}$ (EDR-$\bar{r}$-SJB), combination of BP-HOL and EDR-$\bar{r}$ (EDR-$\bar{r}$-HOL), and the two biased BP schemes based on \emph{expQ} (EDR-$\bar{r}$-expQ and SP-$\bar{r}/(xr)$-expQ) where $\epsilon=0.01$ in~\eqref{E:expQ}. 
Postfix -$\min$ in legend refers to a SP-BP scheme with optimally scaled edge weight $\tilde{\delta}_{e}={\delta}_{e}\bar{r}/(\min_{e\in\ccalE}{\delta}_{e})$.
The details of these BP schemes, including their requirements on local and global information and test rankings, are listed in Table~\ref{tab:list}.

\begin{figure*}[t]
	\centering
	\subfloat[]{
		\includegraphics[height=2.5in]{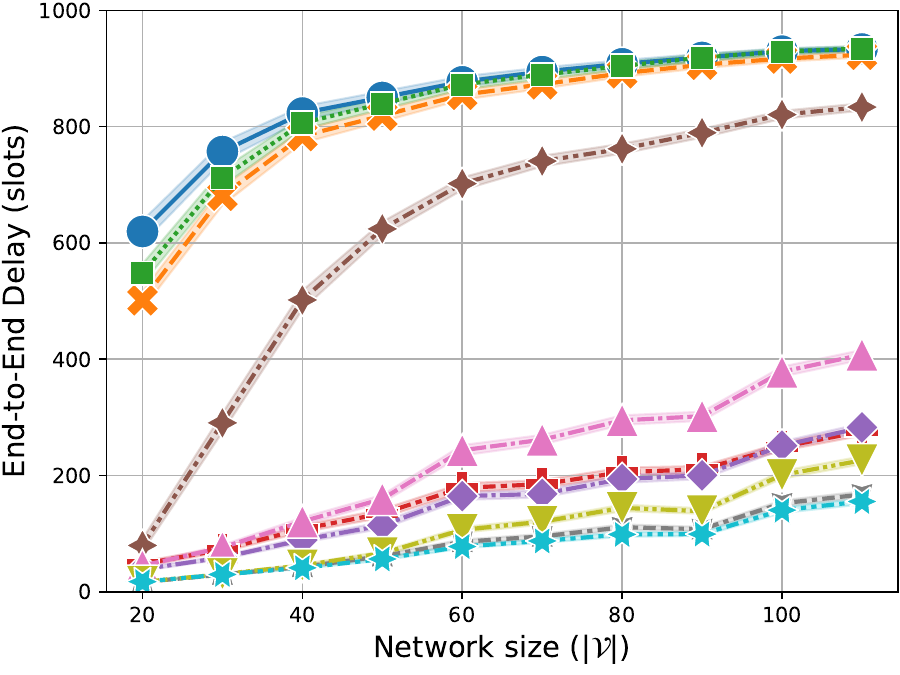}
		\label{fig:regular:delay}\vspace{-0.1in}
	}
	\subfloat[]{
		\includegraphics[height=2.5in]{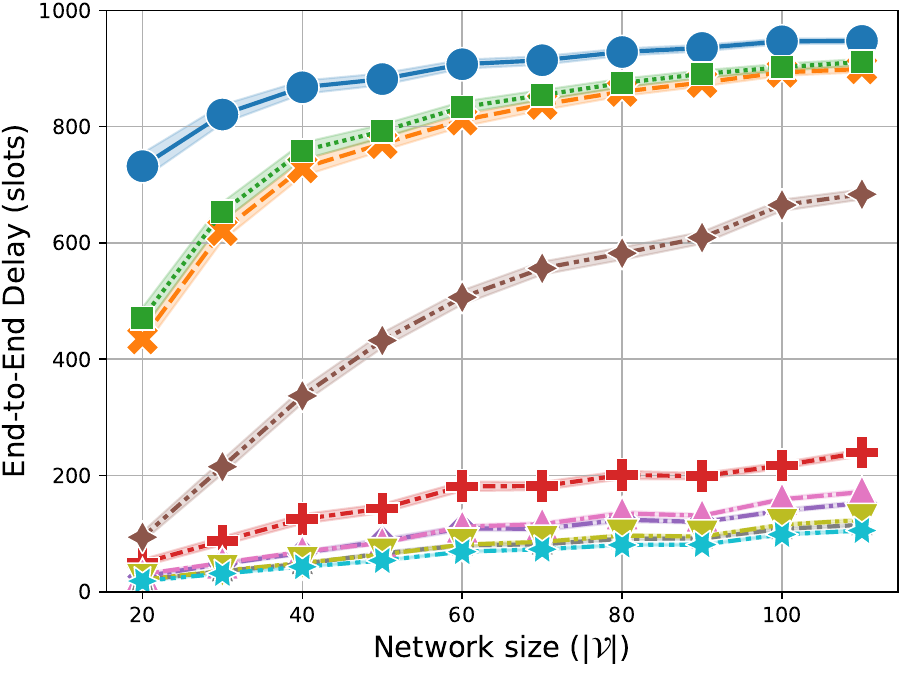}\label{fig:burst:delay:light}\vspace{-0.1in}
	}\\
	\subfloat[]{
		\includegraphics[height=2.5in]{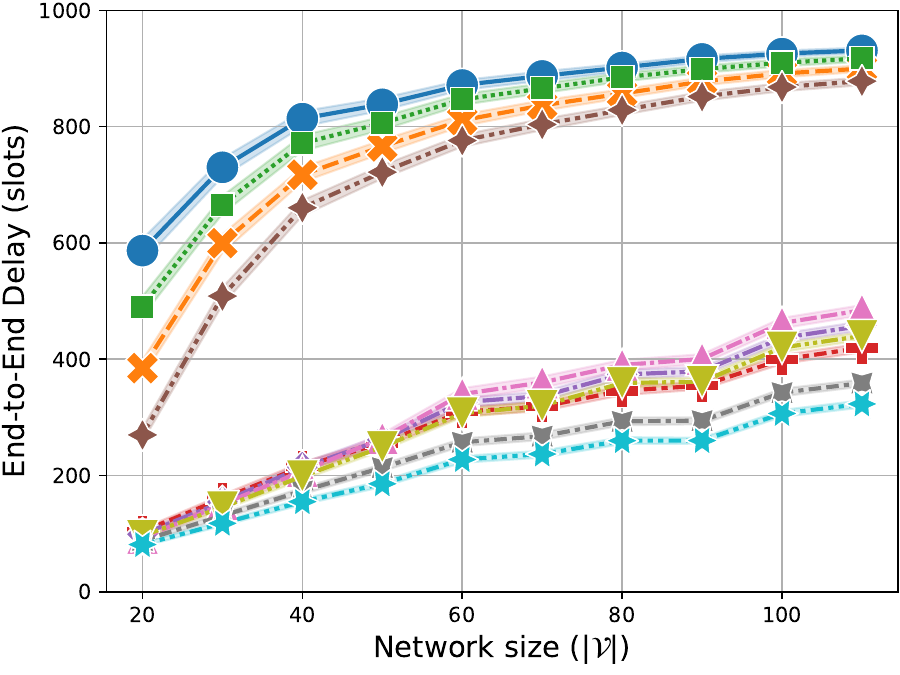}
		\label{fig:burst:delay}\vspace{-0.1in}
	}
	\subfloat[]{
		\includegraphics[height=2.5in]{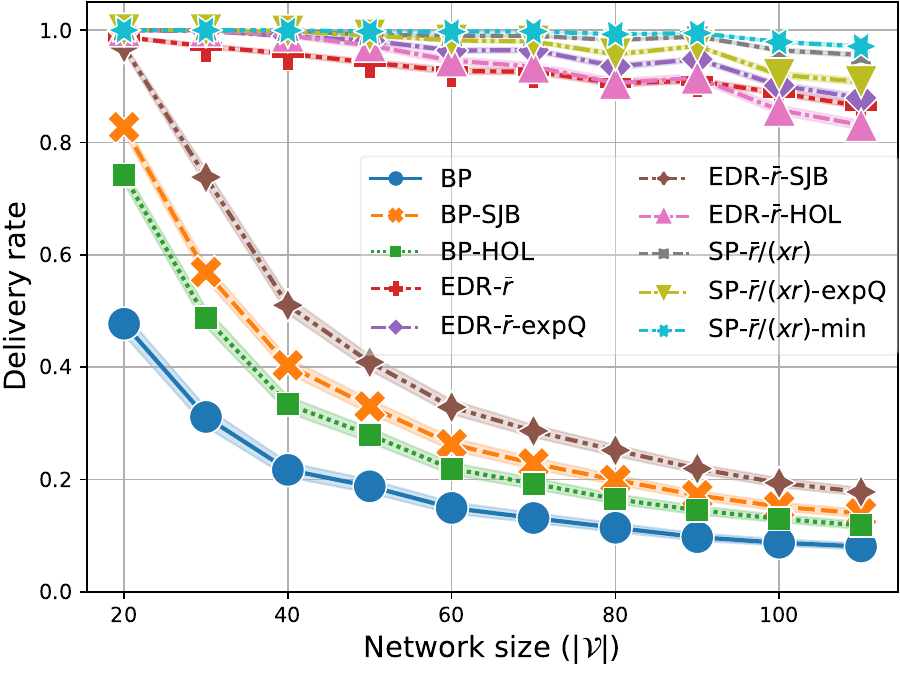}
		\label{fig:burst:delivery}\vspace{-0.1in}
	}
	
	\vspace{-0.1in}
	\caption{{\small Performance of BP algorithms as a function of network size under the unit-disc interference model: 
			(a)~end-to-end delay under streaming traffic, (b)~end-to-end delay under bursty traffic with low flow rate $\lambda(f) \sim \mathbb{U}({2.0,10.0})$ for $t<30$, (c)~end-to-end delay of bursty traffic with high flow rate $\lambda(f) \sim \mathbb{U}({6.6,33.0})$ for $t<30$, and (d)~packet delivery rate of bursty traffic with high flow rate. 
			Simulated with long-term link rates $r_{e} \sim \mathbb{U}(10,42)$, total time steps $T=1000$, and all queues are initialized to be empty. $100$ test instances per point (10 random networks $\times$ 10 realizations of random source-destination pairs and link rates). 
			Error band indicates $95\%$ confidence interval. 
		}
	} 
	\label{fig:main}   
\end{figure*}

\subsection{BIAS SCALING}
\label{sec:results:scaling}

We first examine the optimal bias scaling approach proposed in Section~\ref{sec:solution:scaling}.
We test EDR-$\delta$ and SP-$\bar{r}/(xr)$-min with edge weight adjustment $\underset{e\in\ccalE}{\min}\;\delta_e\vcentcolon=a\bar{r}$, for $a\in\left[0.5,1.5\right]$, on random networks of 100 nodes. 
The results in Fig.~\ref{fig:regular:mpx} show that $a=1.0$ is indeed the optimal setting for EDR-$\delta$ and near-optimal setting for SP-$\bar{r}/(xr)$-min, under both streaming and bursty traffic.
This provides empirical validation of the adjustment rule $\underset{e\in\ccalE}{\min}\;\delta_e\vcentcolon=\bar{r}$ advocated in Section~\ref{sec:solution}.
Notice that the settings tested in Fig.~\ref{fig:regular:mpx} go beyond the simplified setting of CFLP with constant rates treated in Section~\ref{sec:solution}.
However, the proposed scaling is still empirically optimal in these broader scenarios.
Since SP-$\bar{r}/(xr)$-min can tolerate a wide range of $a$, to keep its distributed execution, it can be implemented based on statistical information of edge weights rather than their global minimum.
{The different sensitivities of the end-to-end delay to the scaling of edge weight under the two tested SP-BP schemes, as shown in Fig.~\ref{fig:regular:mpx}, reveal that the shortest path based on link features \cite{zhao2023icassp} not only achieves better performance but is also more robust to the scaling choice.}

\subsection{NETWORK SIZE AND TRAFFIC TYPE}
\label{sec:results:latency}
Next, we evaluate various low-overhead BP algorithms in random networks of different sizes (20-110 nodes), presenting their performances in end-to-end delay and total packet delivery rate after $T=1000$ time steps.
The evaluations are configured with three types of traffic: 1) all streaming flows, 2) all bursty flows, and 3) a mixture of both.

\vspace{-0.2in}
\subsubsection{STREAMING TRAFFIC}\label{sec:results:latency:stream}
When the traffic is configured with all streaming flows, the average end-to-end delays (unit: time slots) as a function of the network size for ten tested BP schemes are presented in Fig.~\ref{fig:regular:delay}.
{Although not presented in Figs.~\ref{fig:main},} the ranking of the BP schemes on packet delivery rate is consistent with the inverted ranking on latency, showing that the latency under streaming traffic is dominated by the contribution of undelivered packets, whose delays are counted as $T$. 
{This is consistent with the results under mixed traffic shown in Table~\ref{tab:list} (Section~\ref{sec:results:latency:mixed}).}

\begin{figure*}[t]
	\centering
	\subfloat[]{
		\includegraphics[height=2.5in]{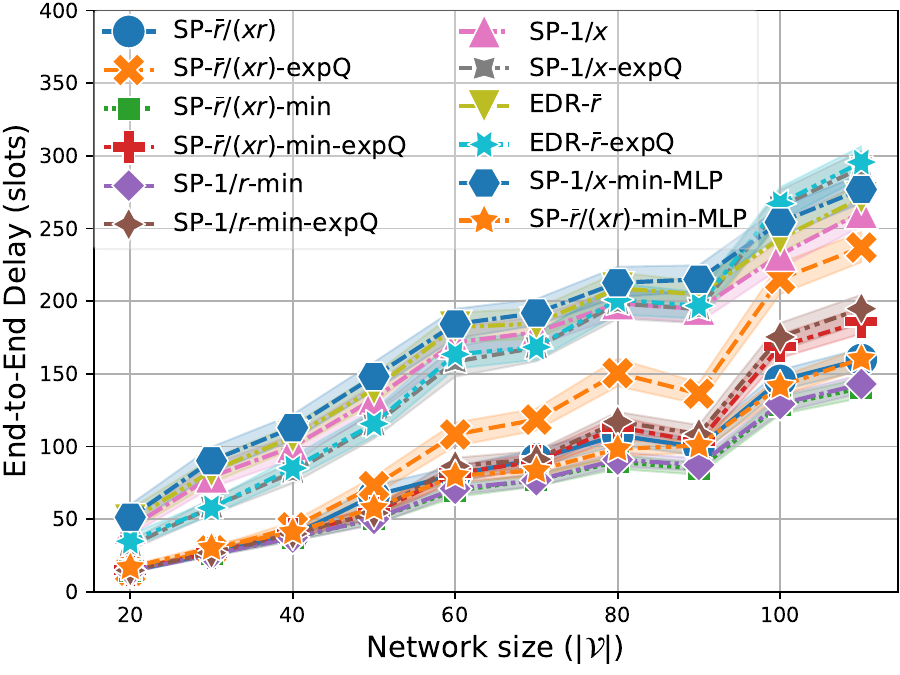}
		\label{fig:mixed:stream:delay}\vspace{-0.1in}
	}
	\subfloat[]{
		\includegraphics[height=2.5in]{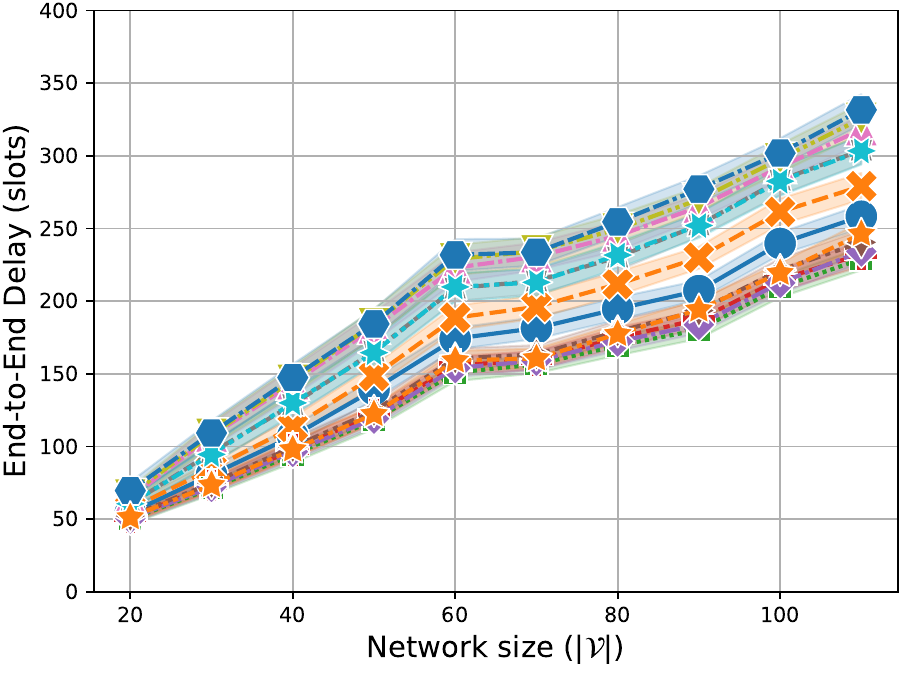}
		\label{fig:mixed:bursty:delay}\vspace{-0.1in}
	} \\
	\subfloat[]{
		\includegraphics[height=2.5in]{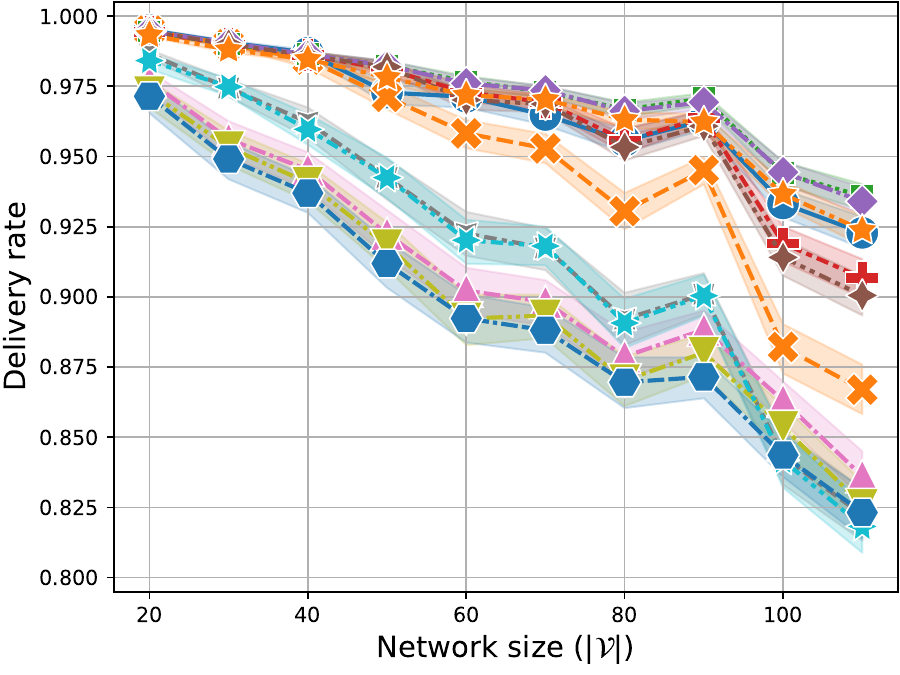}
		\label{fig:mixed:stream:delivery}\vspace{-0.1in}
	}
	\subfloat[]{
		\includegraphics[height=2.5in]{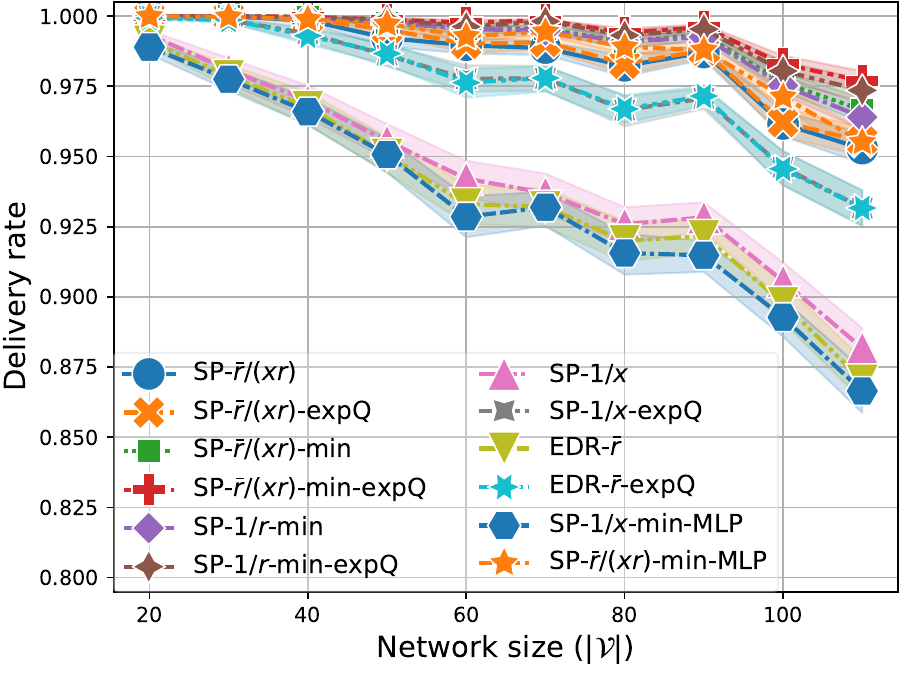}
		\label{fig:mixed:bursty:delivery}\vspace{-0.1in}
	}
	\vspace{-0.1in}
	\caption{{\small Performance of SP-BP schemes as a function of network size with streaming and bursty traffic mixed together at $1:1$ ratio.
			Average end-to-end delays of (a)~streaming flows and (b)~bursty flows. 
			Average packet delivery rates of (c)~streaming flows and (d)~bursty flows.
			Simulated with long-term link rates $r_{e} \sim \mathbb{U}(10,42)$, unit-disc interference model, total time steps $T=1000$, and $100$ test instances per point (10 random networks $\times$ 10 realizations of random flows and link rates).
			Error band indicates $95\%$ confidence interval. 
		}
	} 
	\label{fig:mixed}   
\end{figure*}

While all the BP variations improve the latency over the basic BP, the unbiased BP-SJB and BP-HOL bring limited improvements since they still suffer from the drawbacks of slow startup and random walk.
Their advantages diminish in larger networks, where those drawbacks are more pronounced. 
The biased SP-BP schemes substantially improve the latency, showing that the pre-defined, optimally scaled distance gradients can warm start the system and effectively mitigate random walks.
Queue length-based EDR-$\bar{r}$ achieves a latency lower than 300 time slots, whereas by incorporating additional link features, SP-$\bar{r}/(xr)$ further reduces the latency of EDR-$\bar{r}$ by nearly a half.
Applying optimal bias-scaling to SP-$\bar{r}/(xr)$ (SP-$\bar{r}/(xr)$-min) also slightly improves its latency.

Incorporating delay-based backlog metrics, SJB and HOL, into EDR-$\bar{r}$ can significantly degrade its performance. 
In contrast, \emph{expQ} slightly improves EDR-$\bar{r}$ for smaller networks ($|\ccalV|\leq 90$), and slightly degrades SP-$\bar{r}/(xr)$, 
showing that \emph{expQ} does less harm when applied to streaming traffic, which is not the original purpose of \emph{expQ}. 

\vspace{-0.2in}
\subsubsection{BURSTY TRAFFIC}\label{sec:results:latency:burst}
When the traffic is configured with all bursty flows, the average end-to-end delays (unit: time slots) as a function of the network size for ten tested BP schemes under lighter and heavier traffic loads are respectively presented in Figs.~\ref{fig:burst:delay:light} and~\ref{fig:burst:delay}, and the packet delivery rates corresponding to Fig.~\ref{fig:burst:delay} are illustrated in Fig.~\ref{fig:burst:delivery}.
The flow rates for the lighter and heavier bursty traffic are respectively $\lambda(f) \sim \mathbb{U}(2.0,10.0)$ and $\lambda(f) \sim \mathbb{U}(6.6,33.0)$ for $t<30$, and $\lambda(f)=0$ for $t\geq30$. 
Compared to the cases of streaming traffic, the basic BP performs slightly better under heavier bursty traffic but worse under lighter bursty traffic, 
whereas both SJB and HOL can consistently improve the latency of basic BP regardless of the traffic loads.

Under bursty traffic the shortest path bias is still the most important factor for improved latency and delivery rate, evidenced by the lowest latency achieved by queue length-based SP-$\bar{r}/(xr)$ and SP-$\bar{r}/(xr)$-min under both lighter and heavier traffic loads. 
The fact that they consistently outperform EDR-$\bar{r}$ also demonstrates that link features can better inform shortest path biases than do hop counts.
Optimal bias scaling also benefits bursty traffic, as demonstrated by SP-$\bar{r}/(xr)$-min outperforming SP-$\bar{r}/(xr)$ under heavier traffic loads.

Under lighter bursty traffic, both HOL and \emph{expQ} can significantly improve the latency of EDR-$\bar{r}$ (Fig.~\ref{fig:burst:delay:light}), boosting its delivery rate from $91.0\%$ to $99.8\%$ and $99.5\%$, respectively~\cite{zhao2023enhanced}.
However, as the flow rates triple, they degrade the latency of EDR-$\bar{r}$ (Fig.~\ref{fig:burst:delay}) but still improve its delivery rate (Fig.~\ref{fig:burst:delivery}).
Such an increase in congestion is caused by the fact that HOL and \emph{expQ} can reduce the network capacity of biased BP, as discussed in Section~\ref{sec:results:capacity}.
A similar effect can be observed when \emph{expQ} is applied to SP-$\bar{r}/(xr)$.
Under lighter bursty traffic (Fig.~\ref{fig:burst:delay:light}), SP-$\bar{r}/(xr)$-min, SP-$\bar{r}/(xr)$, and SP-$\bar{r}/(xr)$-expQ achieve the 1st, 2nd, and 3rd best latency, at a $100\%$ delivery rate.
However, under heavier bursty traffic, as shown in Figs.~\ref{fig:burst:delay} and~\ref{fig:burst:delivery}, \emph{expQ} worsens the latency and delivery rate of  SP-$\bar{r}/(xr)$, while optimal bias scaling improves both performances.
For biased BP, \emph{expQ} consistently outperforms HOL and SJB, however, the effectiveness of these delay-based backlog metrics is limited to lighter traffic.

\vspace{-0.2in}
\subsubsection{MIXED TRAFFIC}\label{sec:results:latency:mixed}
We next evaluate low-overhead BP schemes under a mixture of streaming and bursty flows, to assess their effectiveness on the last packet problem under the interference of streaming traffic.
In this case, a source-destination pair is configured as either streaming or bursty flow at equal probabilities.
The bursty flows are configured with the heavier traffic loads described in Section~\ref{sec:results:latency:burst}, to ensure similar numbers of packets are injected into the network for both streaming and bursty flows.
We present the average latency and packet delivery rates of streaming and bursty flows under 15 tested BP schemes, alongside their rankings in Table~\ref{tab:list}.
In general, their rankings are consistent with the results in Sections~\ref{sec:results:latency:stream} and~\ref{sec:results:latency:burst} with similar traffic loads, with the exception that \emph{expQ} is able to improve EDR-$\bar{r}$ and SP-$\bar{r}/(xr)$ for bursty flows in mixed traffic.

\begin{figure}[t!]
\centering
    \includegraphics[width=0.99\linewidth]{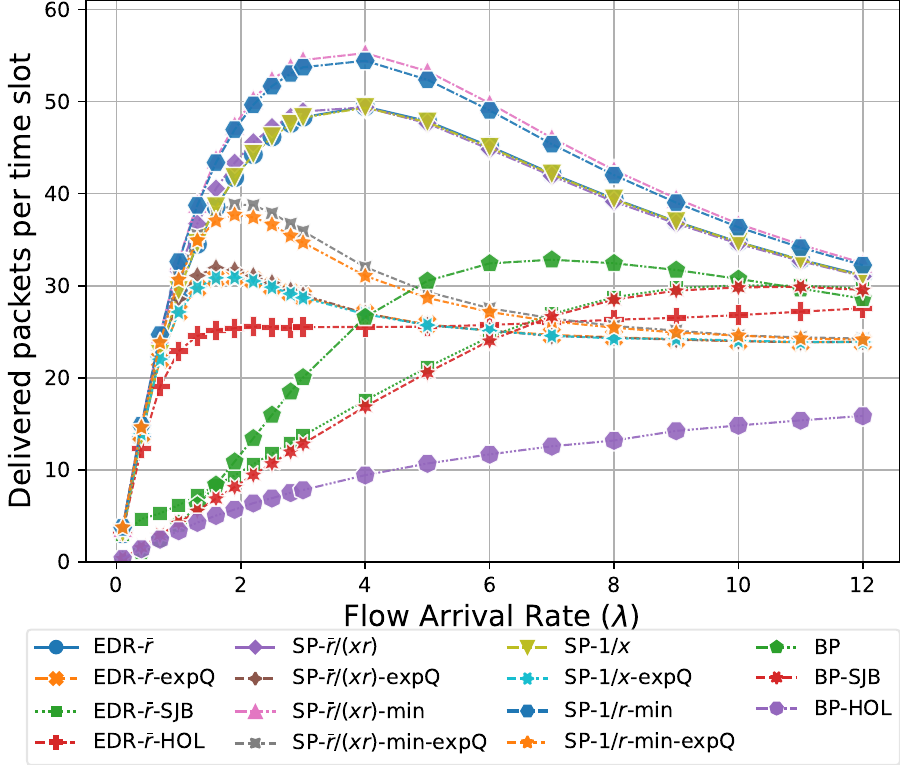}
    \caption{{\small Average end-to-end network throughput (number of packets delivered to their destinations per time slot across the network) under various SP-BP schemes as a function of flow arrival rate. 
    Simulated in random networks of $100$ nodes under unit-disk interference model, long-term link rates $r_{e} \sim \mathbb{U}(10,42)$, total time steps $T=1000$, and $100$ test instances per point (10 random networks $\times$ 10 realizations of random source-destination pairs and link rates). 
    }}
    \label{fig:regular:load}
\end{figure}

We compare the performances of ten SP-BP schemes under mixed traffic, presenting their latency as a function of network size in Figs.~\ref{fig:mixed:stream:delay} and~\ref{fig:mixed:bursty:delay} for  streaming and bursty flows, respectively, and the corresponding packet delivery rates in Figs.~\ref{fig:mixed:stream:delivery} and~\ref{fig:mixed:bursty:delivery}.
Generally, better performance can be achieved by incorporating more local and global information into the shortest path biases.
Notably, SP-$1/x$ can consistently outperform EDR-$\bar{r}$ without access to any global information, whereas the latter requires the global average link rate $\bar{r}$.
When the local long-term link rates $r$ and the global $\bar{r}$ become available, SP-$\bar{r}/(xr)$ can achieve performance comparable to SP-$1/r$-min, which requires additional knowledge of the global $r_{\max}$.
With access to global $(xr)_{\max}$, SP-$\bar{r}/(xr)$-min can achieve the best latency. 
These demonstrate the effectiveness of our proposed method of using GNNs to predict the link duty cycle in biased BP.
The \emph{expQ} consistently improves the delivery rates of biased BP schemes for bursty flows, with the top two delivery rates achieved by SP-$\bar{r}/(xr)$-min-expQ and SP-$1/r$-min-expQ (Fig.~\ref{fig:mixed:bursty:delivery}), demonstrating its effectiveness in solving the last packet problem.

\update{To demonstrate the benefit of adopting GNNs in our enhanced SP-BP schemes, we replace the GCNN with a multi-layer perceptron (MLP) with the same hyperparameters ($L, g_l, \sigma_l$) and training procedure, using node degree as input feature. 
It can be observed from Figs.~\ref{fig:mixed} that the MLP-based schemes, SP-$1/x$-min-MLP and SP-$\bar{r}/(xr)$-min-MLP, always underperform their GCNN-based counterparts due to the inferior ability of MLP in encoding topological information. }

\update{The above results are consistent with similar tests under different node densities and interference levels without retraining the GCNN, as included in our code repository.
It shows that the proposed enhancements generalize well to various scenarios.}

\subsection{NETWORK CAPACITY REGION}
\label{sec:results:capacity}

To illustrate the impact of different low-overhead BP schemes on network capacity, we test them on 10 random networks of size $|\ccalV|=100$ under streaming flows with constant arrival rates ranging from $0.1\leq\lambda\leq12$. 
The total number of packets delivered to their destinations is collected for each test instance after $T=1000$ time steps, and the average end-to-end total throughput of 15 BP schemes as a function of flow arrival rate are presented in Fig.~\ref{fig:regular:load} and the corresponding peak values are listed in Table~\ref{tab:list} (last column). 
For most BP schemes, the end-to-end total throughput first increases with the flow arrival rate, and then decreases slowly after reaching a peak.
When the flow arrival rate is too large, the MaxWeight scheduler will prioritize the links near source nodes, so that fewer packets will reach their destinations, causing a decrease in the end-to-end throughput.

Since the network-wide rate of packet injection $\lambda|\ccalF|$ is linear to the flow rate, the total number of backlogged packets within the network will increase rapidly after the end--to-end throughput reaches its peak. 
Therefore, we consider the peak end-to-end throughput as the network capacity under certain BP schemes. 
The basic BP reaches the network capacity of $32.83$ at $\lambda=7.0$.
With optimally scaled shortest path biases derived from different levels of information access, queue length-based biased BP schemes increase the network capacity of basic BP by $50\%$ to $68\%$ at the same peak of $\lambda=4.0$, while keeping their throughput always above that of basic BP.
Notably, EDR-$\bar{r}$, SP-$1/x$, and SP-$\bar{r}/(xr)$ achieve almost identical maximum throughput, where SP-$\bar{r}/(xr)$ has larger throughput when $\lambda < 4.0$.
The fully optimized SP-$\bar{r}/(xr)$-min and SP-$1/r$-min respectively achieve the best and second network capacities.

The delay-based backlog metrics, including SJB, HOL, and \emph{expQ}, all reduce their network capacity region at different levels, paying a cost for improving delivery rates at light-weight traffic loads.
Our \emph{expQ} keeps achieving higher network capacity than SJB and HOL, when applied to EDR-$\bar{r}$.

\subsection{NODE MOBILITY}
\label{sec:results:mobility}

\begin{table}[!b]
	\renewcommand{\arraystretch}{1.1}
	\caption{End-to-end delay on random networks of $100$ nodes, $T=1000$ under the unit-disk interference model and node mobility. 
	} 
	\label{tab:mobility}
	\centering
    \scriptsize
	\begin{tabular}{|l|l|l|l|l|}
        \hline
	    \multirow{2}{9mm}{Bias update} & \multicolumn{2}{|c|}{EDR-$\bar{r}$} & \multicolumn{2}{|c|}{SP-$\bar{r}/(xr)$}  \\ \cline{2-5}
         & Delay (std.)  & Delivery (std.) & Delay (std.) & Delivery (std.) \\ \hline
        Ideal & 278.3 (106.1) & 81.2\% (8.7\%) & \textbf{155.2} (80.9) & \textbf{90.8\%} (6.1\%) \\ \hline
        Neighbor & 331.8 (103.8) & 75.7\% (9.1\%) & \textbf{282.5} (87.5) & \textbf{78.5\%} (7.8\%) \\ \hline
	\end{tabular}
\end{table}

Lastly, we run 100 test instances on 10 draws of random networks with 100 nodes and node mobility, in which for every 100 time steps, $ 10 $ random nodes take a random step 
on the 2D plane, i.e., $\bbkappa \in \reals^2, \bbkappa_{i} \sim \mathbb{N}(0,0.1) $ (normal distribution with mean of zero and standard deviation of $0.01$), conditioned on keeping the network connected. 
The long-term link rate of a newly created link $e$ is set as $r_{e} \sim \mathbb{U}(10,42)$.
We compare the end-to-end delay and delivery rate of our neighborhood bias update rule in \eqref{E:neighbor} with those of ideal (instantaneous) SSSP in Table~\ref{tab:mobility}.
Under the practical neighborhood update rule, SP-$\bar{r}/(xr)$ suffers a larger loss from the ideal (but impractical) SSSP than EDR-$\bar{r}$, i.e., the delivery rate of SP-$\bar{r}/(xr)$  dropped by $12\%$ compared to a drop of $5.6\% $ with EDR-$\bar{r}$.
Despite larger delays and lower delivery rates, EDR-$\bar{r}$ suffers less under node mobility, \update{showing that it is more robust to transient errors in biases}.

\section{CONCLUSIONS}
\label{sec:conclusions}

In this paper, we studied BP routing and scheduling with queue-agnostic shortest path bias, which effectively addresses the drawbacks of BP while retaining its queue stabilizing capability.
Our four enhancements to SP-BP: delay-aware link weight, optimal bias scaling, low-overhead bias maintenance, and \emph{expQ} (a delay-aware backlog metric that works well with SP biases) are demonstrated to be effective in improving the latency and effective throughput of BP schemes.
The benefits of our enhanced SP-BP schemes are limited to well-connected wireless networks with relatively low mobility.
Further investigations are required to enable BP schemes in wireless networks that are intermittently connected or subject to varying levels of node mobility. 
Other potential extensions to this work include considering the uncertainties and/or shifts in link features and their global minimal/maximal values; facilitating the prediction of the link duty cycles by incorporating additional link and traffic features, such as betweenness centralities and flow arrival rates; exploring other GNN architectures \update{such as Graph Attention Networks}, and \update{developing} training techniques \update{for faster convergence and/or continuous online learning using real-time network data}.

\appendices

\section{PROOF OF INEQUALITY}\label{app:inequality}
To prove $(\max[q-o, b]+p)^2\leq q^2+o^2+p^2+b^2+2bp+2q(p-o)$ for $q\geq0, o\geq0, b\geq 0, p\geq 0$, consider two cases:

Case 1: $q-o \geq b$, we have
\begin{align*}
    (\max[q-o, b]+p)^2 = & (q-o+p)^2 \\
    = & q^2 + o^2 + p^2 + 2qp - 2qo - 2op \\
    \leq & q^2 + o^2 + p^2 + 2q(p - o) \\
    \leq & q^2 + o^2 + p^2 + b^2 + 2bp + 2q(p - o), 
\end{align*}
since $2op\geq 0$, $b^2 + 2bp \geq 0$. 

Case 2: $q-o < b$, we have
\begin{align*}
    (\max[q-o, b]+p)^2 = & (b+p)^2 \\
    = & b^2 + p^2 + 2bp \\
    \leq & (q-o)^2 + 2qp + b^2 + p^2 + 2bp \\
    = & q^2 + o^2 + p^2 + b^2 + 2bp + 2q(p - o), 
\end{align*}
since $(q-o)^2 + 2qp \geq 0$.

With Case 1 and Case 2, we prove the inequality.

\update{
\section{JUSTIFICATION OF LIGHTWEIGHT TRAFFIC SETTING}\label{app:frate}
For all streaming flow setting, given the test setting in Section~\ref{sec:results}, in each test instance, the expected number of source nodes is $0.4|\ccalV|$, each with an average arrival rate of $0.6=(0.2+1.0)/2$, the average exogenous packet arrival rate per node is computed as $ 0.24=0.6\times 0.4 $.}

\update{The expected link rate of a bidirectional link is $26=(10+42)/2$. 
With a node density of $8/\pi$, the expected number of immediate neighbors of each node is $8$. 
Given the average conflict degree of each link as $34.6$, the average maximum outflow rate per node is computed as $3\approx(26\times 8\times 0.5)/34.6$, assuming an even split of link rate for both directions.
}

\bibliographystyle{ieeetr}

\bibliography{strings,refs}

\input{biography.tex}

\end{document}

%% file: biography.tex
\begin{IEEEbiography}
[{\includegraphics[width=1in,height=1.25in,clip,keepaspectratio]{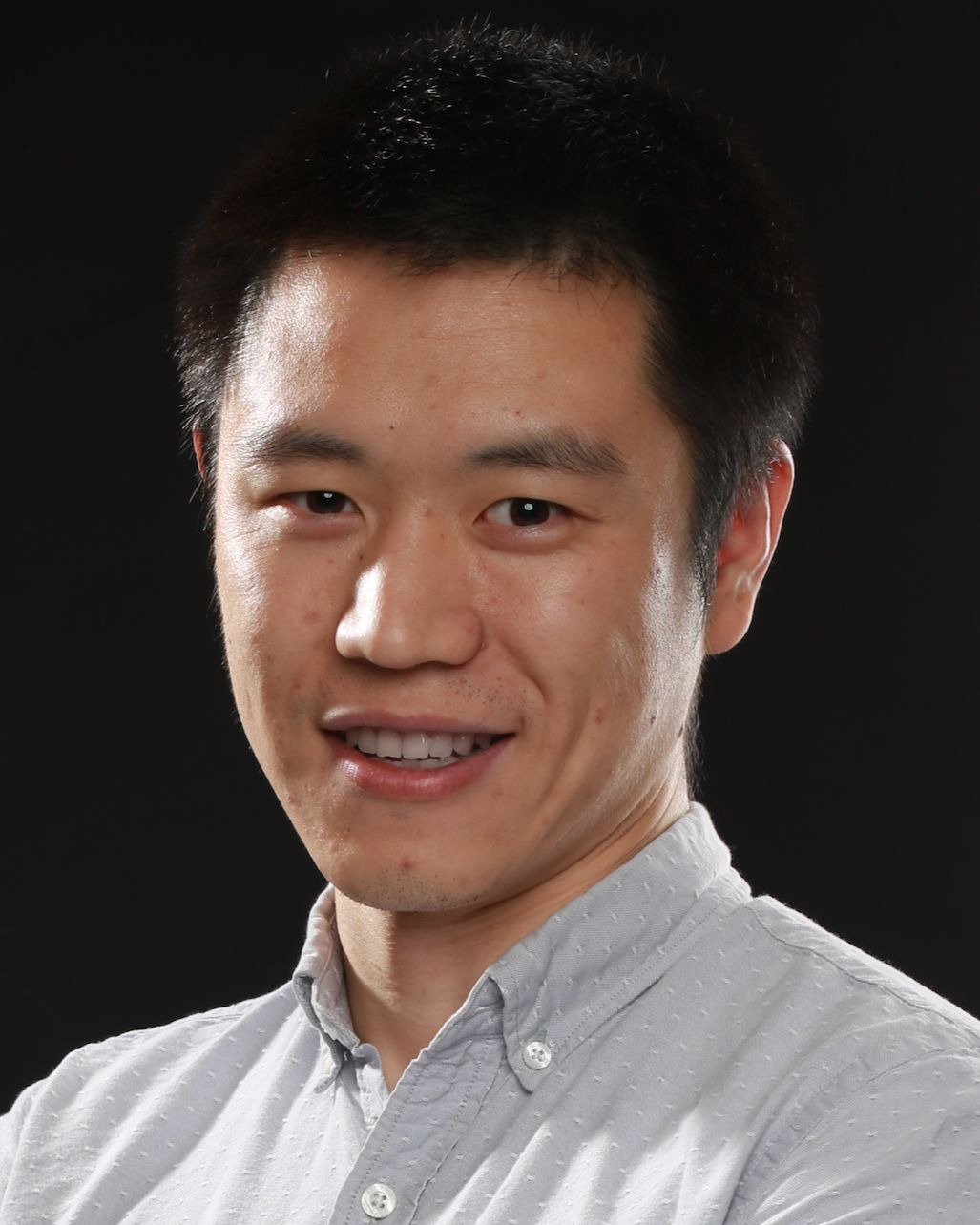}}]{Zhongyuan Zhao} (S'13--M'18) received his B.Sc. and M.S. degrees in Electronic Engineering from the University of Electronic Science and Technology of China, Chengdu, China, in 2006 and 2009, respectively. He received his Ph.D. degree in Computer Engineering from the University of Nebraska-Lincoln, Lincoln, NE, in 2019, under the guidance of Prof. Mehmet C. Vuran. From 2009 to 2013, he worked for ArrayComm and Ericsson, respectively, as an engineer developing 4G base-station. He joined Rice University in 2019 as a postdoctoral research associate at the Department of Electrical and Computer Engineering, advised by Prof. Santiago Segarra. Currently, he is a Research Assistant Professor at Rice University. Dr. Zhao’s current research interests focus on machine learning and signal processing for wireless communications and networking.
\end{IEEEbiography}

\begin{IEEEbiography}
[{\includegraphics[width=1in,height=1.25in,clip,keepaspectratio]{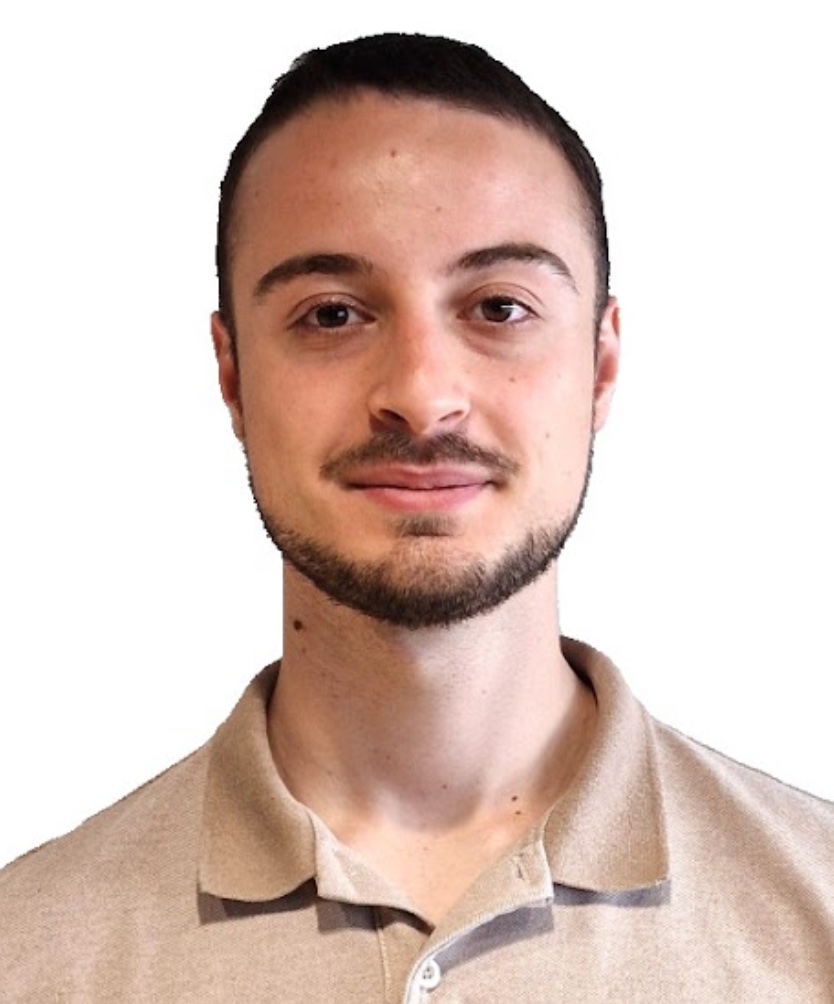}}]{Bojan Radojičić}
received the Bachelor's degree in Information Engineering from the Faculty of Technical Sciences, University of Novi Sad, Novi Sad, Serbia in 2023, where he is currently a master student in Engineering Management.
During his studies, he won two first places in the Artificial Intelligence competition organized by BEST in Serbia.
He was a Visiting Research Fellow at Rice University in 2022, and did internships at Prime Software Serbia, Serbia in 2023, KION group in Frankfurt, Germany in 2024, in the fields of Research, Finance, Business and Data Science.
He currently works as an Associate Product Manager Intern at Google in London, UK.
\end{IEEEbiography}

\begin{IEEEbiography}
[{\includegraphics[width=1in,height=1.25in,clip,keepaspectratio]{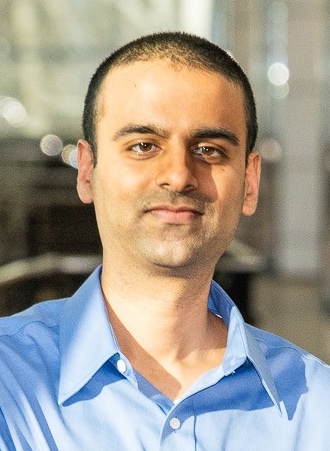}}]{Gunjan Verma}
received the B.S. degree in mathematics, computer science, and economics from Rutgers University, the M.A. degree in computational biology from Duke University, and the M.S. degree in mathematics from Johns Hopkins University. He is currently a Computer Scientist with the U.S. Army Research Laboratory (ARL), Adelphi, MD, USA. His research interests include Bayesian statistics and machine learning and their application to problems in networking, communications, and robotics.
\end{IEEEbiography}


\begin{IEEEbiography}
[{\includegraphics[width=1in,height=1.25in,clip,keepaspectratio]{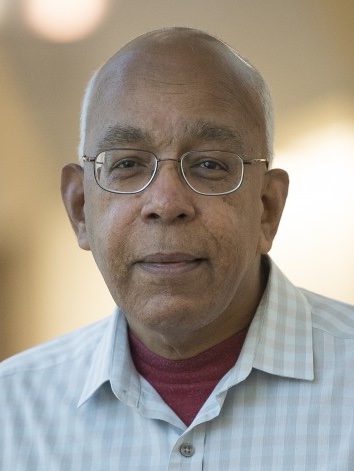}}]{Ananthram Swami}
(Life Fellow, IEEE) received the B.Tech. degree from IIT-Bombay, the M.S. degree from Rice University, and the Ph.D. degree from the University of Southern California (USC), all in electrical engineering. He is currently with the U.S. Army's DEVCOM Army Research Laboratory (ARL) as the Army’s Senior Research Scientist (ST) for Network Science, and is an ARL fellow. Prior to joining ARL, he held positions with Unocal Corporation, USC, and CS-3. He has held visiting faculty positions at INP, Toulouse, and Imperial College, London.  His recent awards include the 2018 IEEE ComSoc MILCOM Technical Achievement Award and the 2017 Presidential Rank Award (Meritorious).
\end{IEEEbiography}

\begin{IEEEbiography}
[{\includegraphics[width=1in,height=1.25in,clip,keepaspectratio]{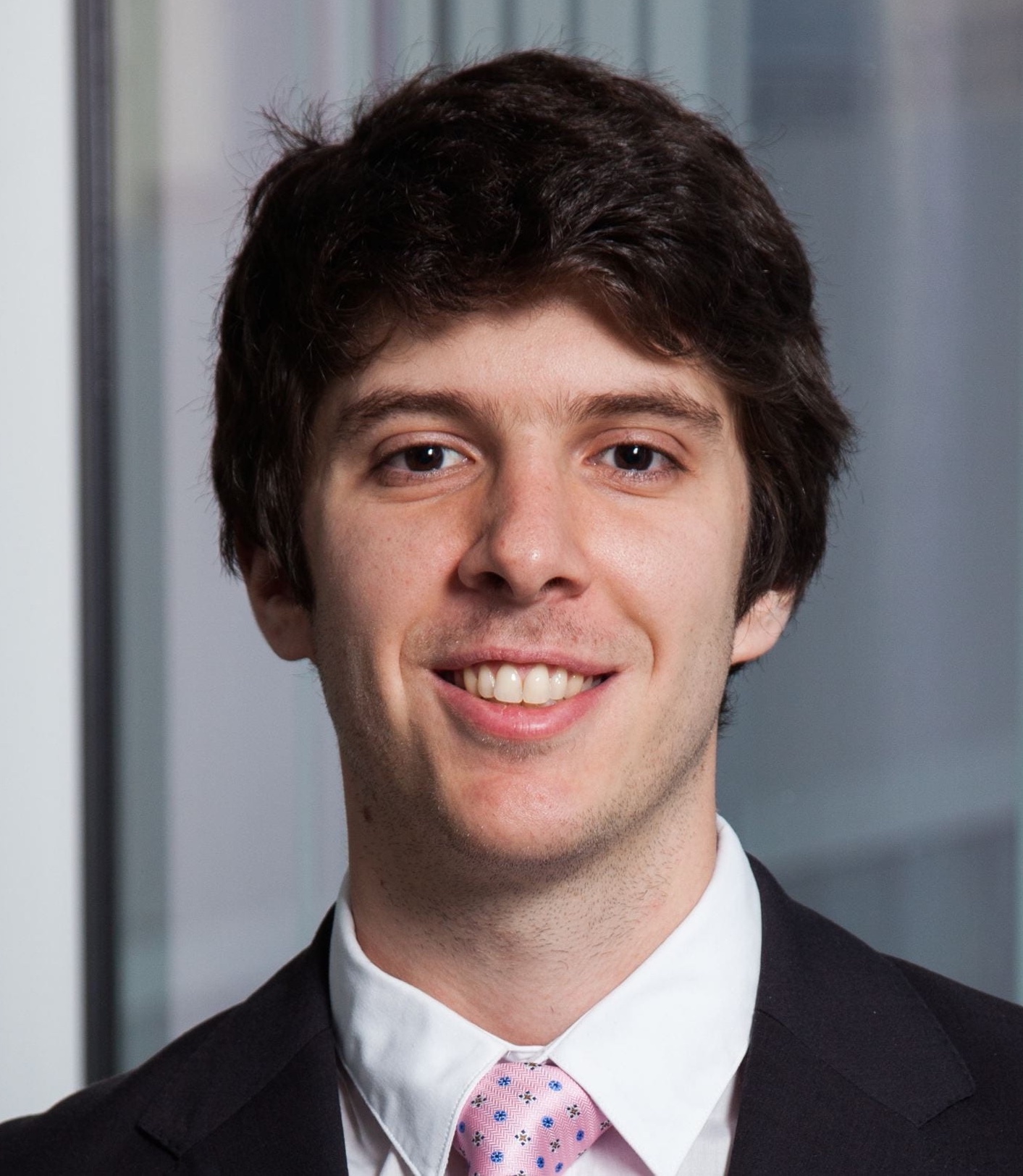}}]{Santiago Segarra} (Senior Member, IEEE) received the B.Sc. degree (Hons.) (Valedictorian) in industrial engineering from the Instituto Tecnológico de Buenos Aires (ITBA), Argentina, in 2011, the M.Sc. in electrical engineering from the University of Pennsylvania (Penn), Philadelphia, in 2014 and the Ph.D. degree in electrical and systems engineering from Penn in 2016. From September 2016 to June 2018 he was a postdoctoral research associate with the Institute for Data, Systems, and Society at the Massachusetts Institute of Technology. 
In 2018, Dr. Segarra joined the Department of Electrical and Computer Engineering, Rice University as an Assistant Professor, and is currently an Associate Professor. He also holds courtesy appointments with the Departments of Computer Science and Statistics. His research interests include network theory, data analysis, machine learning, and graph signal processing. He received the 2011 Outstanding Graduate Award granted by the National Academy of Engineering of Argentina, the 2017 Penn’s Joseph and Rosaline Wolf Award for Best Doctoral Dissertation in Electrical and Systems Engineering, the 2020 IEEE Signal Processing Society Young Author Best Paper Award, the 2021 Rice’s School of Engineering Research + Teaching Excellence Award, three early career awards (NSF CAREER, ARO ECP, and ARI Early Career), and five best conference paper awards. Since 2022, he has dedicated part of his time consulting for Microsoft Research.
\end{IEEEbiography}